\documentclass[12pt]{article}

\usepackage{amsmath,amssymb,amsthm,amscd}
\usepackage{graphicx}

\usepackage{natbib}
\usepackage{algorithm}	% those two packages are helpful if you use a lot of algorithms in your thesis
\usepackage{algorithmic}
\usepackage{bm}
\usepackage{bbm}
\usepackage{enumitem}

\usepackage{cprotect}

\usepackage[toc,page]{appendix}

\graphicspath{ {Images/} }

\theoremstyle{plain}

\newtheorem{Theorem}{Theorem}[section]		% those commands are for numbering Theorems etc within the chapter

\newtheorem{Lemma}[Theorem]{Lemma}

\newcommand{\ts}{^{\sf T}} %transposition

\newcommand{\ilpdif}[2]{\partial #1/{\partial #2 }}
\newcommand{\pdif}[2]{\frac{\partial #1}{\partial #2}}
\newcommand{\pddif}[3]{\frac{\partial^2 #1}{\partial #2 \partial #3}}

\newcommand{\ilpddif}[3]{\partial^2 #1/{\partial #2 \partial #3}}
\newcommand{\bp}{{\bm \beta}}

\newcommand{\blind}{1}

\begin{document}

\def\spacingset#1{\renewcommand{\baselinestretch}%
{#1}\small\normalsize} \spacingset{1}

%%%%%%%%%%%%%%%%%%%%%%%%%%%%%%%%%%%%%%%%%%%%%%%%%%%%%%%%%%%%%%%%%%%%%%%%%%%%%%

\if1\blind
{
  \title{\bf Fast calibrated additive quantile regression}
  \author{Matteo Fasiolo\thanks{
    This work was funded by EPSRC grant by EPSRC grants EP/K005251/1, EP/N509619/1 and the first author was also partially supported by EDF. }, \hspace{.2cm}
    Simon N. Wood, \\
    School of Mathematics, University of Bristol \\
    Margaux Zaffran,  \\
    %\vspace{0.4cm}
    ENSTA Paris \\ %\vspace{0.4cm} \\
    and \\
    Rapha\"el Nedellec and Yannig Goude \\
    %\vspace{0.4cm}
    {\'E}lectricit{\'e} de France R\&D}
    \date{}
  \maketitle
} \fi

\if0\blind
{
  \bigskip
  \bigskip
  \bigskip
  \begin{center}
    {\LARGE\bf Fast calibrated additive quantile regression}
\end{center}
  \medskip
} \fi

\bigskip
\begin{abstract}
We propose a novel framework for fitting additive quantile regression models, which provides well calibrated inference about the conditional quantiles and fast automatic estimation of the smoothing parameters, for model structures as diverse as those usable with distributional GAMs, while maintaining equivalent numerical efficiency and stability. The proposed methods are at once statistically rigorous and computationally efficient, because they are based on the general belief updating framework of \cite{bissiri2016general} to loss based inference, but compute by adapting the stable fitting methods of \cite{wood2016smoothing}. We show how the pinball loss is statistically suboptimal relative to a 
novel smooth generalisation, which also gives access to fast estimation methods. Further, we provide a novel calibration method for efficiently selecting the `learning rate' balancing the loss with the smoothing priors during inference, thereby obtaining reliable quantile uncertainty estimates. Our work was motivated by a probabilistic electricity load forecasting application, used here to demonstrate the proposed approach. The methods described here are implemented by the \verb|qgam| R package, available on the Comprehensive R Archive Network (CRAN).
\end{abstract}

\noindent%

%\newpage
%\spacingset{1.5} % DON'T change the spacing!

%\newpage

%\tableofcontents

\section{Introduction} \label{sec:introduction}

Generalized Additive Models \citep[GAMs,][]{hastie1990generalized} are flexible and interpretable statistical models that are widely used in applied statistics, especially since the advent of efficient and stable methods for smoothing parameter selection and interval estimation in this model class \cite[see e.g.][]{wood2000modelling, ruppert2003semiparametric, kim2004smoothing, fahrmeir2004penalized, wood2006generalized}. The purpose of this work is to provide an equivalently useful framework for well-calibrated additive quantile regression models. Our methods are novel in that all smoothing parameters and other hyper parameters are estimated automatically using numerically robust and efficient methods which produce uncertainty estimates simultaneously with point estimates.

We were motivated by problems in electricity load forecasting. Electricit\'e de France (EDF), France's main electricity producer, has had considerable success using conventional GAMs for operational load forecasting. However, the whole conditional load distribution is rarely needed for production planning purposes, which focus mostly on tail estimates. This is because the loss function associated with forecasting errors is highly asymmetric, due to technical constraints (e.g. plant-specific start-up times or increasing fuel cost along the electricity production stack) and to the regulatory framework (e.g. monetary sanctions for over/under production). Further, the conditional distribution of the electricity load is typically highly skewed and time-dependent. At system-wide or substation level this problem is relatively mild, but new technologies (e.g. smart meters) are producing datasets where this issue is much more extreme, due to the low level of aggregation. Full distributional modelling of the response distribution might be overly ambitious for these upcoming applications, hence it might be preferable to focus on estimating only the conditional quantiles most relevant to production planning or smart grid management.

To be usable in practical forecasting, additive quantile regression methods must have several properties: 1) the range of model structures available for modelling quantiles must be comparable to that available under conventional GAMs, otherwise the benefits of modelling quantiles may be offset by insufficient model flexibility; 2) smoothing and other tuning parameters must be selected automatically, otherwise the modelling process becomes too labour intensive and subjective for  operational use; 3) uncertainty estimation has to be part of model estimation, since knowing forecast uncertainty is essential for operational use and 4) methods must be sufficiently numerically efficient and robust for routine deployment. The work reported here started when two of the authors (YG and RN) were participating in the GEFCom2014 forecasting competition, and found that existing additive quantile regression method implementations failed to meet these requirements, forcing them to develop the ad hoc procedure described in \cite{gaillard2016additive}.

The framework developed in this paper meets the four requirements by taking an empirical Bayesian approach to the general belief-updating framework of \cite{bissiri2016general}. Specifically we represent smooth relationships between regressors and the quantile of interest using spline basis expansions, and impose Gaussian smoothing priors to control model complexity. Random effects and parametric terms present no extra complication. By adopting a statistically improved smooth generalisation of the usual quantile regression `pinball' loss \citep{koenker2005quantile}, we are able to perform the computations required for belief updating of priors using the loss, and to estimate smoothing parameters, using the computational methods for general smooth modelling of \cite{wood2016smoothing}. This allows us to achieve properties 1-4, provided that we can obtain the additional `learning rate' parameter required by the general belief updating framework. We show how to do this efficiently and automatically in order to achieve good calibration of the uncertainty estimates. Figure \ref{fig:nicePlots} provides some simple examples of the variety of models that our approach encompasses.

This is an advance relative to existing methods because, to our knowledge, pre-existing additive quantile regression methods fail to meet one or more of the four practically important requirements set above. For instance, the \verb|quantreg| R package, which is based on the methods of \cite{koenker2013quantreg}, %only permits additive models whose smooth terms are at most bi-dimensional, and it 
requires users to select the smoothing parameters manually. The gradient boosting quantile regression method implemented by the \verb|mboost| R package \citep{hothorn2010model} %does not limit the dimensionality of the smooth terms, but it 
requires users to manually choose the degrees of freedom used by each base model. In addition, \verb|mboost| uses bootstrapping to estimate parameter uncertainty, while the approach proposed here quantifies uncertainty using computationally efficient asymptotic approximations. %without adding to the leading order computational cost. 
\cite{yue2011bayesian} and \cite{waldmann2013bayesian} describe how to perform Bayesian inference for semi-parametric additive quantile regression models. The first proposal is implemented in the \verb|INLA| software \citep{martins2013bayesian}, but the associated documentation discourages its use. In the second fitting is performed only via Markov Chain Monte Carlo methods, which are much slower than the direct optimisation methods proposed here. Further, it does not produce credible intervals with adequate frequentist properties for extreme quantiles, as mentioned below. The \verb|vgam| R package \citep{yee2008vgam} provides a method for fitting additive quantile regression models, but also in this case the complexity of the smooth terms is determined manually. The work of \cite{lin2013variable} is not an alternative to what we propose here, because their focus is variable selection, rather than smoothing.

Quantile regression is traditionally based on the pinball loss \citep{koenker2005quantile}, and not on a distributional model for the observations density, $p(y|{\bf x})$, which impedes direct application of Bayes's rule. To circumvent this problem, \cite{yu2001bayesian} propose adopting an Asymmetric Laplace (AL) model for $p(y|{\bf x})$, due to the equivalence between the AL negative log-density and the pinball loss. While \cite{sriram2013posterior} prove that the resulting posterior concentrates around the true quantile, naively treating the AL density as an adequate probabilistic description of the data is problematic. In particular, \cite{waldmann2013bayesian} show that the resulting posterior credible intervals have poor frequentist calibration properties, especially for tail quantiles. Furthermore, this work will demonstrate that, in a non-parametric setting, selecting
the scale parameter of the AL density using a likelihood based approach can lead to inaccurate quantile estimates (see Section \ref{sec:addExample}). We solve both issues by adopting the beliefs updating framework of \cite{bissiri2016general}, and by coupling it with a calibration method which explicitly aims at achieving good frequentist properties.

This work also addresses the limitations implied by direct use of the pinball loss. The first issue is that this loss is piecewise linear, which impedes the use of computationally efficient fitting methods, designed to work with continuously differentiable, strongly convex functions. \cite{yue2011bayesian} and \cite{oh2012fast} address this problem by proposing smooth approximations to, respectively, the AL density and the pinball loss. The second issue (see Section \ref{sec:approxLapl}) is that the pinball loss is statistically suboptimal relative to a smoothed generalisation of the loss. Rather than smoothing the loss as little as possible (as in previous work) we therefore adopt the novel approach of using the degree of loss smoothness that minimizes the asymptotic MSE of the model regression coefficients.
%

%Similarly to \cite{oh2012fast}, we provide results concerning the impact of the approximation on the estimated quantiles. Further, we exploit them to develop a scale invariant parametrisation of the coefficients that controls the trade-off between accuracy and computational stability.

%The advantage of this approach is that it results in a normalised approximation to the AL density, which is critical for the purpose of model selection. 

%In addition, when implementing quantile regression using the extended GAM method, described in Section \ref{sec:exGam}, we encountered numerical instabilities, due to the curvature of the new log-density being close to zero in the tails (the AL log-density has no curvature at all). This required modifying the routines used to estimate the regression coefficients and to select the smoothing parameters, as described in Section \ref{sec:zeros1} and \ref{sec:zeros2}.

\begin{figure} 
\centering
\includegraphics[scale=0.44]{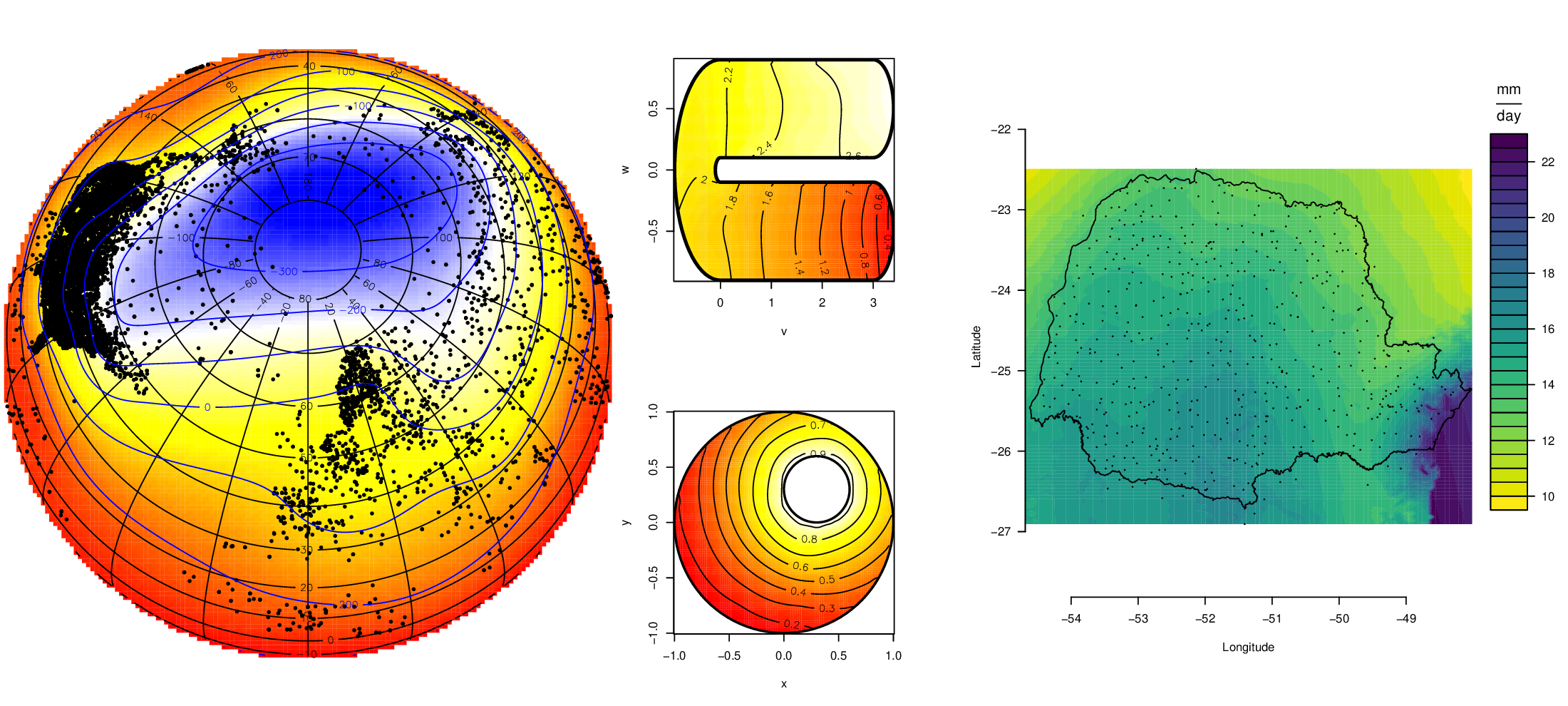} 
\cprotect\caption{Examples of the smooth components that may be included in the additive quantile regression models fitted using the approach proposed here. Left: effect of spatial location, defined using splines on the sphere, on quantile $\tau = 0.1$ of minimum daily temperatures, estimated using the Global Historical Climatology Network (GHCN) dataset \citep{menne2012overview}. The Gulf Stream is visible. Centre: finite area spatial components, based on soap film smoothers, of two GAM fits for $\tau = 0.5$. The data is simulated. Right: sum of the effects of spatial location, defined using an isotropic thin-plate spline basis, distance from the ocean and elevation, on quantile $\tau = 0.9$ of average weekly rainfall in Paran\'{a} state, Brazil. The dataset is available within the  \verb|R-INLA| R package \citep{lindgren2015bayesian}.}
\label{fig:nicePlots}
\end{figure}

%The framework of \cite{bissiri2016general} introduces an extra parameter, the learning rate, which determines the relative weight of the prior and the loss within the posterior. In the context of quantile regression, this parameter is completely confounded with the scale of the AL density. The scale parameter of the new loss plays a similar role and it is important to select it carefully, because it controls both the wiggliness of the fitted quantile and the width of posterior credible intervals. Here the learning rate is selected automatically by matching the credible intervals of the conditional quantiles, derived using an asymptotic approximation of the posterior, to their empirical distribution. This approach leads to robustness in smoothing parameter selection, and to credible intervals that roughly achieve nominal coverage levels. In contrast, \cite{waldmann2013bayesian} show that, if the AL density is adopted naively, it is difficult to achieve adequate coverage, even when MCMC methods are used.

%This mis-specification problem can cause posterior credible intervals for model parameters not to achieve nominal frequentist coverage, even asymptotically \citep{kleijn2012bernstein}. 

The rest of the paper is organised as follows. In Section \ref{sec:background} we briefly review additive quantile regression based on the pinball loss, while in Section \ref{sec:approxLapl} we describe the new loss function and show how its tuning parameter can be selected. In Section \ref{sec:fixedSigma} we show how quantile regression can be set in a Bayesian context using the framework of \cite{bissiri2016general}, and we explain how additive quantile regression models can be fitted efficiently, if the new loss function is adopted. We propose a novel approach for posterior calibration in Section \ref{sec:tuningLearn}, and we test it on simulated examples in Section \ref{sec:simulExamples}. In Section \ref{sec:gefcom} we demonstrate the performance of the proposed approach for probabilistic electricity load forecasting.

\section{Background on additive quantile regression} \label{sec:background}

Quantile regression aims at modelling the $\tau$-th quantile (where $\tau \in (0, 1)$) of the response, $y$, conditionally on a $p$-dimensional vector of covariates, $\bm x$.  More precisely, if $F(y|\bm x)$ is the conditional c.d.f. of y, then the $\tau$-th conditional quantile is
$
\mu = \inf \{y:F(y|{\bm x}) \geq \tau\}.
$
The $\tau$-th conditional quantile can also be defined as the minimiser of the expected loss
\begin{equation} \label{eq:trueObj}
L(\mu|{\bm x}) = \mathbb{E} \big \{ \; \rho_\tau(y-\mu) | {\bm x} \big \}  = \int \; \rho_\tau(y-\mu) dF(y|{\bm x}),  
\end{equation}
w.r.t. $\mu = \mu({\bm x})$, where  
\begin{equation} \label{eq:checkF}
\rho_\tau(z) = (\tau-1)z \mathbbm{1}(z < 0) + \tau z \mathbbm{1}(z \geq 0),
\end{equation}
is the so-called pinball loss. Given a sample of size $n$, one approximates $dF(y)$ with its empirical version, $dF_n(y)$, which leads to the quantile estimator
$$
\hat{\mu} = \underset{\mu}{\text{argmin}}\; \frac{1}{n} \sum_{i=1}^n \; \rho_\tau\{y_i-\mu( {\bm x}_i)\},
$$
where ${\bm x}_i$ is the $i$-the vector of covariates.

In this work we assume that $\mu({\bm x})$ has an additive structure such as
$
\mu({\bm x}) = \sum_{j=1}^m f_j({\bm x}),  
$
where the $m$ additive terms can be fixed, random or smooth effects, defined in terms of spline bases. For instance, a marginal smooth effect could be
$
f_j({\bm x}) = \sum_{k=1}^r \beta_{jk} b_{jk}({x}_{j}),
$
where $\beta_{jk}$ are unknown coefficients and $b_{jk}({x}_{j})$ are known spline basis functions. Analogous expressions can be used to define joint or more complex smooths, such as those shown in Figure \ref{fig:nicePlots}. The basis dimension $r$ is typically chosen to be sufficiently generous that we can be confident of avoiding over-smoothing, but the actual complexity of $f_j$ is controlled by a penalty on ${\bm \beta}_{j}$, designed to penalise departure from smoothness. 

More specifically, let $\mu({\bm x}_i) = {\bf x}_{i}\ts{\bm \beta}$, where ${\bf x}_i$ is the $i$-th row of the $n \times d$ design matrix $\bf X$, containing the spline basis functions evaluated at $\bm x_i$, and define the penalised pinball loss
\begin{equation}\label{eq:penPinLoss}
V({\bm \beta}, {\bm \gamma}, \sigma) =  \sum_{i=1}^n  \frac{1}{\sigma} \rho_\tau\{y_i-\mu( {\bm x}_i)\} + \frac{1}{2}\sum_{j=1}^m \gamma_j {\bm \beta}\ts {\bf S}_{j} {\bm \beta},
\end{equation}
where $\bm \gamma = \{\gamma_1, \dots, \gamma_m\}$ is a vector of positive smoothing parameters and the ${\bf S}_j$'s are positive semi-definite matrices, used to penalise the wiggliness of $\mu(\bm x)$. $1/\sigma > 0$ is the so-called `learning rate', which determines the relative weight of the loss and the penalty. As we will show later, it is possible to let $\sigma$ depend on $\bm x$, which can be advantageous when the variance of $y$ varies strongly with the covariates. For fixed $\sigma$ and $\bm \gamma$, the regression coefficients can be estimated by minimising (\ref{eq:penPinLoss}) and in Section \ref{sec:fixedSigma} we will show that the resulting estimator, $\hat{\bm \beta}$, can be seen as a Maximum A Posteriori (MAP) estimator under the belief updating framework of \cite{bissiri2016general}. While the main challenge is the selection of $\sigma$ and $\bm \gamma$, we first demonstrate that the pinball loss is statistically suboptimal, and can be improved in a manner that has the useful side effect of facilitating computation. 

\section{Optimally smoothing the pinball loss} \label{sec:approxLapl}
 
We consider the following generalisation of the scaled pinball loss
\begin{equation} \label{eq:elfLoss}
\tilde{\rho}(y-\mu) = (\tau - 1)\frac{y-\mu}{\sigma} + \lambda \log(1+e^{\frac{y-\mu}{\lambda\sigma}}),
\end{equation}
where $\lambda>0$ and the pinball loss is recovered as $\lambda \rightarrow 0$. Let $\text{Beta}(\cdot, \cdot)$ be the beta function. While the proposed loss can be seen as an instance of the function smoothing methods of \cite{chen1995smoothing}, we call (\ref{eq:elfLoss}) the Extended Log-F (ELF) loss because, upon normalisation, $\exp(- \tilde{\rho})$ becomes
\begin{equation} \label{eq:logFLambda}
\tilde{p}_{F}(y - \mu)=\frac{e^{(1-\tau)\frac{y-\mu}{\sigma}}(1+e^{\frac{y-\mu}{\lambda\sigma}})^{-\lambda}}{\lambda\sigma\text{Beta}\big\{\lambda(1-\tau),\lambda\tau\big\}},
\end{equation}
which is the p.d.f. of an extension of the log-F density of \cite{jones2008class}, as detailed in the Supplementary Material   \ref{sec:exLogDens} (henceforth SM \ref{sec:exLogDens}). 

Let $\hat{\bm \beta}$ be the minimiser of the penalised ELF loss, the latter being defined analogously to (\ref{eq:penPinLoss}). Define ${\bf V} = \tau(1-\tau)\mathbb E({\bf x}_{i}{\bf x}_{i}\ts)$ and $u_i = y_i - \mu({\bm x}_i)$, its p.d.f $f_{u|{\bm x}}(u|{\bm x})$ and c.d.f. $F_{u|{\bm x}}(u|{\bm x})$. Assume that the pairs $\{{\bm x}_i, {y}_i\}$ are i.i.d., that $\lambda \rightarrow 0$ as $n\rightarrow\infty$ and consider a simplified setting where the smoothing penalty is asymptotically dominated by the ELF loss and the true quantile $\mu^0(\bm x) \approx {\bf x}\ts \bm \beta_0$, for some $\bm \beta_0$. The latter assumption implies that the true quantile is approximately linear in the design vector $\bf x$ of evaluated spline bases, so that $d$ does not grow with $n$.  In SM \ref{app:OptSmoothLoss} we show that, under further assumptions specified therein, the asymptotic MSE of $\hat{\bm \beta}$ is
\begin{align} \label{eq:AMSE}
\text{AMSE}(h) = & \mathbb{E}\{n(\hat{\bm \beta}-{\bm \beta}_0)(\hat{\bm \beta}-{\bm \beta}_0)\ts\} \nonumber \\ = & 
{\bm \Sigma}_f^{-1}{\bf V}^{\frac{1}{2}}[{\bf I} + nh^4{\bf B}{\bf B}\ts - h{\bf A}{\bf A}\ts]{\bf V}^{\frac{1}{2}}{\bm \Sigma}_f^{-1} + O(h^2) + O(1/nh) + O(nh^6),
\end{align}
where $h = \lambda\sigma$, ${\bm \Sigma}_f = \mathbb{E}\{{\bf x}_i{\bf x}_i \ts f_{u|{\bm x}}(0|{\bm x}_i)\}$, ${\bf B} = \pi^2\mathbb{E}\{ f'_{u | {\bm x}}(0 | {\bm x}_{i}){\bf V}^{-1/2}{\bf x}_{i}\}/6$ and ${\bf A} = \mathbb{E}\{f_{u|{\bm x}}(0|{\bm x}_i)^{1/2}$\\${\bf V}^{-1/2}{\bf x}_i\}$. Following \cite{kaplan2017smoothed}, we seek the value of $h$ that minimises (\ref{eq:AMSE}), after having discarded the higher-order terms. This is justified for $h = O(n^{-\psi})$ with $1/5<\psi<2/5$, under which the first term on the r.h.s. of (\ref{eq:AMSE}) dominates the remainder. The AMSE of a linear combination $\sqrt{n}{\bf c}(\hat{\bm \beta} - {\bm \beta})$ is ${\bf u}\ts({\bf I} + nh^4{\bf B}{\bf B}\ts - h{\bf A}{\bf A}\ts){\bf u}$, where ${\bf u} = {\bf V}^{1/2}{\bm \Sigma}_f^{-1} {\bf c}$. Given that there is no single $h^*$ which minimises the AMSE for all values of ${\bf c}$, we minimise $\text{tr}(nh^4{\bf B}{\bf B}\ts - h{\bf A}{\bf A}\ts)$, which is equivalent to minimising the sum of eigenvalues of $nh^4{\bf B}{\bf B}\ts - h{\bf A}{\bf A}\ts$. This leads to
$$
h^* = \left[\frac{{\bf A}\ts{\bf A}}{4n{\bf B}\ts{\bf B}}\right]^{\frac{1}{3}},
$$
which, under the further assumption that the $u_i$'s are independent of the ${\bm x}_i$'s, becomes
\begin{equation} \label{eq:optBanZ}
\tilde{h}^* = \left[\frac{d}{n}\frac{9f_{u}(0)}{\pi^4f'_{u}(0)^2}\right]^{\frac{1}{3}},
\end{equation}
where $d = \text{dim}(\bf x)$, as proved in SM \ref{app:OptSmoothLoss}. So clearly the pinball loss ($\lambda=0$) is not optimal, and we should use the ELF loss with smoothness determined by $\tilde{h}^*$. This substitution of the smoothed loss greatly simplifies computation as it permits the use of smooth optimisation methods for estimation. We now turn to the practical estimation of $\tilde h^*$.

\subsection{Estimating $\tilde h^*$} \label{sec:estHstar} 

Here we propose methods for estimation of $\tilde h^*$, considering also the treatment of heteroscedastic data. In particular, using a single bandwidth is inadequate in contexts where the variance of $y$ strongly depends on $\bm x$. To address this issue, consider the location-scale model $y|{\bm x} \sim \alpha({\bm x}) + \kappa({\bm x}) z$, where $\mathbb{E}(z|{\bm x})=0$ and $\text{var}(z|{\bm x})=1$. Assume that the $z_i$'s are i.i.d, and let $\tilde{h}^*_z$ be the optimal bandwidth for regressing $z$ on $\bf x$. Then the corresponding optimal bandwidth for $y$ is $\tilde{h}^*(\bm x) = \tilde{h}_z^* \kappa({\bm x})$, as can be verified using a simple change of variable argument. %XXX To verify this, you need to go into the derivation of, for instance, E(W) and write h = h k(x) and the use the substitution u  = (y - mu(x)) / k(x). Then the end you will get a single scalar h^* and then you get h^*(x) = h^* k(x) because of how you defined h to start with. 
Given that in our context $\tilde{h}_z^* \kappa({\bm x}) = \lambda\sigma$, one of the terms on the r.h.s. must depend on $\bm x$. We choose $\sigma(\bm x) = \sigma_0 \tilde{\sigma}(\bm x)$, where $\sigma_0$ has been chosen using the methods of Section \ref{sec:tuningLearn} and $n^{-1}\sum_{i=1}^n\tilde{\sigma}({\bm x}_i)=1$. This implies that $\lambda = n^{-1}\tilde{h}_z^*\sum_{i=1}^n\kappa({\bm x}_i)/\sigma_0$. Under the above location-scale model, this leads to $1/\sigma(\bm x) \propto \omega(\bm x) = f_u(0|\bm x)$, which is the optimal weight function under the pinball loss $\sum_{i=1}^n \omega({\bm x}_i)\rho_{\tau}\{y_i - \mu(\bm x_i)\}$ \citep{koenker2005quantile}. Hence, while the baseline learning rate $1/\sigma_0$ is determined using the methods of Section \ref{sec:tuningLearn}, its $\bm x$-dependent component $1/\tilde{\sigma}({\bm x})$ is proportional to the optimal weighting function for quantile regression. It is reasonable to expect that the relative learning rate will be near-optimal also under the ELF loss, which is a close approximation to the pinball loss.

Our approach to loss bandwidth selection can be summarized in the following steps:
\begin{enumerate}
\item estimate $\alpha(\bm x)$ and $\kappa(\bm x)$. Here we do this using a Gaussian GAM, where the mean and variance are estimated jointly using the methods of \cite{wood2016smoothing}. We model $\alpha(\bm x)$ using the same model used for $\mu_\tau(\bm x)$, while for $\kappa(\bm x)$ we typically use a simpler model (see Section \ref{sec:gefcom} for an example). %% Canditate4141
\item Obtain the standardised residuals $z_i = \{y_i - \hat{\alpha}(\bm x_i)\}/\hat{\kappa}(\bm x_i)$, for $i = 1, \dots, n$, and get estimates $\hat{f}_z(\hat{\xi}_\tau)$ and $\hat{f}'_z(\hat{\xi}_\tau)$, where $\hat{\xi}_\tau$ is an estimate of the $\tau$-th quantile of $z$. We do this parametrically, by fitting the flexible sinh-arsinh distribution of \cite{jones2009sinh} to the $z_i$'s. If $|\hat{\xi}_\tau - \hat{\xi}_{\tau^*}| < \epsilon$, where $0 < \epsilon \ll 1 $ and $\hat{\xi}_{\tau^*}$ is the mode of $\hat{f}_z$, we set $\hat{\xi}_\tau$ to $\hat{\xi}_{\tau + \delta}$ if $\hat{\xi}_\tau - \hat{\xi}_{\tau^*} > 0$ (or to $\hat{\xi}_{\tau - \delta}$ if $\hat{\xi}_\tau - \hat{\xi}_{\tau^*} < 0$), for some small $\delta > 0$. This is done to avoid dividing by $\hat{f}'_z(\hat{\xi}_\tau) \approx 0$ in (\ref{eq:optBanZ}).
\item Get the bandwidth $\tilde{h}^*(\bm x) = \tilde{h}_z^*\hat{\kappa}(\bm x)$, where $\tilde{h}_z^*$ is obtained by plugging $\hat{f}_z$ and $\hat{f}'_z$ in (\ref{eq:optBanZ}), and by setting $d$ to be the number of effective degrees of freedom used to model $\alpha(\bm x)$ in step 1. Decompose the bandwidth into $\lambda$ and $\sigma(\bm x)$ as explained above.
\end{enumerate}
Notice that the Gaussian GAM and the sinh-arcsinh density of steps 1 and 4 have to be fitted only once, so the resulting estimates can then be used to determine the bandwidths to be used for quantile regression at several probability levels.  
 
\section{Model fitting with known learning rate} \label{sec:fixedSigma}

Having defined a smooth generalisation of the pinball loss, and proposed methods for selecting its degree of smoothness, we describe a framework for fitting splines based additive quantile models. We first explain how to estimate the regression coefficients and the smoothing 
parameters, given $\sigma_0$. Estimation of $\sigma_0$ is covered in Section \ref{sec:tuningLearn}.

\subsection{Bayesian quantile regression via coherent belief-updating} \label{sec:bayesQR}

To set quantile regression in a Bayesian framework, we need to define a prior distribution on the regression coefficients, $p(\bm \beta)$, and a mechanism for updating it to the corresponding posterior, $p({\bm \beta}|{\bf y})$. Here we use the smoothing prior $\bm \beta \sim N({\bf 0}, {\bf S}^-)$, where ${\bf S}^-$ is an appropriate generalised matrix inverse of matrix ${\bf S}^{\bm \gamma} = \sum_{i=1}^m\gamma_i{\bf S}_j$, where $\bm \gamma$ and the ${\bf S}_j$'s have been defined in Section \ref{sec:background}. Given such a prior, direct application of Bayes' rule is impeded by the fact that we base quantile regression on the ELF loss, not on a probabilistic model for the observation density, $p(y|{\bm \beta})$, so the likelihood function is missing. Fortunately, this obstacle can be overcome by adopting the general belief-updating framework of \cite{bissiri2016general}, within which a prior distribution can be updated to produce a posterior while using a loss function, rather than a full likelihood, to connect model parameters to the data. Before applying it to quantile regression, we briefly outline the framework in its general form.

Assume that we are interested in finding the vector of model parameters $\bm \beta$ minimising
\begin{equation} \label{eq:exloss}
\mathbb{E}\{L(\bm \beta)\} = \int L({y}, \bm \beta) f({y}) d{y},
\end{equation}
where $L(\cdot, \cdot)$ is a general loss function and $f({y})$ is the p.d.f. of ${y}$. Suppose that we have a prior belief about $\bm \beta$, quantified by the prior density $p(\bm \beta)$. Then \cite{bissiri2016general} argue that, given some data ${y}$, a coherent approach to updating $p(\bm \beta)$ is the posterior
$$
p(\bm \beta|y) = \frac{e^{-\frac{1}{\sigma}L({y}, \bm \beta)} p(\bm \beta)}{\int e^{-\frac{1}{\sigma}L({y}, \bm \beta)} p(\bm \beta) \, d{\bm \beta}},
$$ 
where $1/\sigma > 0$ is a `learning rate', determining the relative weight of the loss and the prior. When multiple samples, ${\bf y} = \{ y_1, \dots, y_n \}$, are available this becomes
\begin{equation} \label{eq:scaledGibbsPost}
p(\bm \beta|{\bf y}) \propto  e^{-\frac{1}{\sigma}\sum_{i = 1}^n L(y_i, \bm \beta)} p(\bm \beta).
\end{equation}
where $\sum_{i = 1}^nL(y_i, \bm \beta)$ is an estimate of (\ref{eq:exloss}). Following \cite{syring2015scaling} we call (\ref{eq:scaledGibbsPost}) the `Gibbs posterior' and its negative normalising constant the `marginal loss'. 

Quantile regression, which we base on the ELF loss, fits squarely into this framework. In fact, the Gibbs posterior corresponding to such loss is
\begin{equation} \label{eq:ELFPost}
p(\bm \beta|{\bf y}) \propto \prod_{i=1}^n \tilde{p}_{F}\{y_i - \mu({\bm x}_{i})\}\, p(\bm \beta),
\end{equation}
where $\tilde{p}_{F}$ is the ELF density (\ref{eq:logFLambda}), which implicitly depends on $\tau$, $\lambda$, $\sigma$ and $\bm \beta$ (the latter via $\mu({\bm x}_i)$). 
%As the examples will show, selecting $\sigma$ correctly is of critical importance for obtaining accurate quantile fits and uncertainty estimates. In strongly heteroscedastic settings, we let $\sigma$ be a function of the covariates via the decomposition $\sigma(\bm x_i) = \sigma_0\tilde{\sigma}(\bm x_i)$, as detailed in Section \ref{sec:approxLapl}. 
In the next section we show how the regression coefficients can be estimated by maximizing the Gibbs posterior (\ref{eq:ELFPost}), given the smoothing parameters and learning rate.

%%
%\begin{equation} \label{eq:gamlssForm}
%\sigma({\bf x}) = \sigma_0 \exp \Big \{ \sum_{j=1}^m f_j({\bf x}) \Big \},
%\end{equation}
%%
%where $\sigma_0$ is the reciprocal of the learning rate. This is selected using an outer iteration, which approximately calibrates the posterior marginal density of the conditional quantile, $\mu({\bf x})$. The additive terms $f_j(\cdot)$ are fixed, random or smooth effects, and their purpose is to modulate the learning rate, so that the speed of learning is inversely proportional to the variability of $y|{\bf x}$. For fixed $\sigma_0$, the effects $f_j(\cdot)$ are determined using the efficient methods described in Section \ref{sec:fixedSigma}.

%Here explain how to estimate the regression coefficients and the smoothing parameters, assuming $\lambda$ and $\sigma(\bm x) = \sigma_0 \tilde{\sigma}(\bm x)$ to be fixed. In particular, for fixed $\sigma_0$,   we assume that $\sigma_0$ has been selected using the calibration method of Section \ref{sec:tuningLearn} and that, given $\sigma_0$, $\lambda$ and $\tilde{\sigma}(\bm x)$ have been chosen using the methods of Section \ref{sec:approxLapl}.

\subsection{Estimating the regression coefficients, $\beta$, 
given $\gamma$ and $\sigma_0$} \label{sec:estRegrCoef}

Indicate with $\text{lo}\{\mu(\bm x_i), \sigma(\bm x_i)\}$ the $i$-th element of the ELF loss (\ref{eq:elfLoss}) where, for fixed $\sigma_0$, parameters $\lambda$ and $\sigma(\bm x_i) = \sigma_0 \tilde{\sigma}(\bm x_i)$ have been selected using the methods of Section \ref{sec:approxLapl}.  Then, the negative Gibbs posterior log-density of $\bm \beta$ is proportional to the penalised loss
\begin{equation} \label{eq:LogPosterior}
\tilde{V}({\bm \beta},{\bm \gamma}, \sigma_0) = \sum_{i=1}^n \text{lo}\{\mu(\bm x_i), \sigma(\bm x_i)\} + \frac{1}{2}\sum_{j=1}^m \gamma_j {\bm \beta}\ts {\bf S}_{j} {\bm \beta}.
\end{equation}
Hence, MAP estimates of the regression coefficients, $\hat{\bm \beta}$, can be obtained by minimising (\ref{eq:LogPosterior}), for fixed $\bm \gamma$ and $\sigma_0$. Given that the objective function is smooth and convex, this could be done efficiently using Newton algorithm, but a more stable solution can be obtained by exploiting orthogonal methods for solving least squares problems. In particular, notice that the minimiser of (\ref{eq:LogPosterior}) corresponds to that of 
\begin{equation} \label{eq:DevCrit}
\tilde{V}_D(\bm \beta, \bm \gamma, \sigma_0) = \sum_{i=1}^n \text{Dev}_i\left\{\bm \beta, \sigma({\bm x}_i)\right\} + \sum_{j=1}^m \gamma_j \bm \beta\ts {\bf S}_j \bm \beta,
\end{equation}
where $\text{Dev}_i\left\{\bm \beta, \sigma({\bm x}_i)\right\} = 2[\text{lo}\{\mu(\bm x_i), \sigma(\bm x_i)\}-\tilde{\text{ll}}]$ and $\tilde{\text{ll}}$ are, respectively, the $i$-th component of the model deviance, based on (\ref{eq:logFLambda}), and the saturated loss, obtained by minimising (\ref{eq:elfLoss}) w.r.t. $\mu$. Then the regression coefficients can be estimated by Penalised Iteratively Re-weighted Least Squares (PIRLS), that is by iteratively minimising
\begin{equation} \label{eq:pirls}
\sum_{i=1}^n w_i\{z_i - \mu_i\}^2 + \sum_{j=1}^m \gamma_j \bm \beta\ts {\bf S}_j \bm  \beta,
\end{equation} 
where
$$
z_i = {\mu}_i - \frac{1}{2w_i}\frac{\partial \text{Dev}_i}{\partial {\mu}_i}, \;\;\;
w_i = \frac{1}{2}\frac{\partial^2 \text{Dev}_i}{\partial{\mu}^2_i},
$$
while ${\mu}_i = {\bf x}_{i} \ts \bm \beta$ and $\text{Dev}_i = \text{Dev}_i\left\{\bm \beta, \sigma({\bm x}_i)\right\}$.

\subsection{Selecting the smoothing parameters, $\gamma$, given $\sigma_0$} \label{sec:estSmoothPar}

%To select $\bm \gamma$ for fixed $\sigma_0$, one might naively consider minimising the expected loss
%\begin{equation} \label{eq:badLoss}
%C({\bm \gamma},\sigma_0) = \mathbb{E}_{\bm \gamma}\bigg[ \sum_{i=1}^n \text{lo}\{\mu({\bf x}_i), \sigma({\bf x}_i)\} \bigg],
%\end{equation}
%w.r.t. $\bm \gamma$, where the expectation is taken w.r.t. the prior, $\bm \beta \sim N({\bf 0}, {\bf S}^-)$. However, recall that ${\bf S}^-$ is positive semi-definite, because the parametric effects in $\mu({\bm x})$ and $\sigma({\bf x})$ are typically left unpenalised. Hence, integral (\ref{eq:badLoss}) is not convergent in general. 

A natural approach to selecting $\bm \gamma$, for fixed $\sigma_0$, is minimising the marginal loss
\begin{equation} \label{eq:goodLoss}
G({\bm \gamma}, \sigma_0) = - \int \exp \bigg[ -\sum_{i=1}^n \text{lo}\{\mu(\bm x_i), \sigma(\bm x_i)\}  \bigg] p(\bm \beta|{\bm \gamma})d\bm \beta,
%= - \mathbb{E}_{\bm \gamma}\bigg( \exp \bigg[ -\sum_{i=1}^n \text{lo}\{\mu({\bf x}_i), %\sigma({\bf x}_i)\} \bigg] \bigg) 
\end{equation}
which, as we noted in Section \ref{sec:bayesQR}, is the negative of the normalising constant of the Gibbs posterior. This is important from a computational point of view, because $G({\bm \gamma},\sigma_0)$ can be computed and minimised using efficient methods, originally developed to handle marginal likelihoods. In particular, $G({\bm \gamma},\sigma_0)$ involves an intractable integral which can be approximated using a Laplace approximation. This results in the Laplace Approximate Marginal Loss (LAML) criterion
\begin{equation} \label{eq:LAMLsimple}
G_{L}({\bm \gamma}, \sigma_0) = \frac{1}{2}\tilde{V}_D(\hat{\bm \beta}, \bm \gamma, \sigma_0) + n\tilde{\text{ll}} +  \frac{1}{2} \Big [ \log |{\bf X\ts  W  X +  S^{\bm \gamma}| - \log|S^{\bm \gamma}|_+} \Big] - \frac{M_p}{2} \log(2 \pi),
\end{equation}
where $\hat{\bm \beta}$ is the minimiser of (\ref{eq:LogPosterior}), $\tilde{\text{ll}}$ is the saturated loss, $\bf W$ is a diagonal matrix such that ${\bf W}_{ii} = w_i$, $M_p$ is the dimension of the null space of ${\bf S}^{\bm \gamma}$ and $|\bf S^{\bm \gamma}|_+$ is the product of its non-zero eigenvalues.

LAML can be efficiently minimised w.r.t. $\bm \gamma$, using an outer Newton algorithm. Numerically stable formulas for computing LAML and its derivatives are provided by \cite{wood2016smoothing}. Importantly, the derivatives of $\hat{\bm \beta}$ w.r.t. $\bm \gamma$ are obtained by implicit differentiation which requires computing derivatives up to fourth order of the ELF loss w.r.t. $\mu$. Notice that the $w_i$'s in (\ref{eq:pirls}) and (\ref{eq:LAMLsimple}) can be very close to zero when fitting quantile regression models based on the ELF density, hence obtaining reliable and numerically stable estimates requires modifying the PIRLS iteration and the computation of (\ref{eq:LAMLsimple}) and its derivatives. This more stable implementation is described in SM \ref{app:stableNewDensity}.

The Laplace approximation to the negative marginal log-likelihood based on the ELF density is obtained simply by adding $\sum_i \log[\lambda\sigma(\bm x_i)\text{Beta}\{\lambda(1-\tau),\lambda\tau\}]$ to (\ref{eq:LAMLsimple}), and it is possible to optimise it w.r.t. $\sigma_0$ as well as ${\bm \gamma}$. But $\sigma_0$ is confounded with the learning rate, so this can not be justified by the \cite{bissiri2016general} framework. Indeed, in Section \ref{sec:simulExamples} we refer to this approach as LAML selection of $\sigma_0$ and present examples of its failure in practice: it often produces inaccurate fits and poor interval calibration. Instead, Section \ref{sec:tuningLearn} presents a calibration-based approach to the selection of $\sigma_0$ which, as the examples will show, alleviates both issues.

\section{Calibrating $\sigma_0$} \label{sec:tuningLearn}

Here we propose a novel method for selecting $\sigma_{0}$, which aims at obtaining approximately well-calibrated credible intervals for the quantile function, $\mu({\bm x})$. In particular, let $C_{\alpha}\{\sigma_{0},{\bf y}\}$ be the credible interval for $\mu({\bm x})$, at level $\alpha\in(0,1)$. The objective is selecting $\sigma_{0}$ so that 
\begin{equation}
\mathbb{P}\big[\mu^{0}({\bm x})\in C_{\alpha}\{\sigma_{0},{\bf y})\}\big]\approx\alpha,\label{eq:coverProp}
\end{equation}
for all $\alpha$, where $\mathbb{P}$ is the objective probability measure, based on
the data-generating process, and $\mu^{0}({\bm x})$ is the true conditional quantile. % We focus on calibrating intervals for $\mu({\bm x})$, rather than
%for $\bm{\beta}$ or $\sigma({\bf x})$, as this is more practically
%relevant in the context of non-parametric quantile regression. Informally, if we assume that ${\bf x}\in\mathbb{R}^{p}$, then the aim is for (\ref{eq:coverProp}) to hold across the subset of $\mathbb{R}^{p}$ which is sufficiently well covered by design points. 

Let $\hat{\bm{\beta}}$ be the MAP estimate of the regression coefficients and define the covariance matrices ${{\bf V}}=(\bm{\mathcal I}+{\bf S}^{\bm{\gamma}})^{-1}$ and $\tilde{{\bf V}}=(\bm{\mathcal I}\bm{\Sigma}_{\nabla}^{-1}\bm{\mathcal I}+{\bf S}^{\bm{\gamma}})^{-1}$, where $\bm{\mathcal I}$ is the Hessian of the unpenalised loss and $\bm{\Sigma}_{\nabla}=\text{cov}[\nabla_{\bm \beta}\text{lo}\{\mu({\bm x}), \sigma({\bm x})\}|_{\bm \beta = \hat{\bm \beta}}]$ w.r.t. $\mathbb{P}$. We select $\sigma_{0}$ by minimizing
\begin{equation} \label{eq:IKLestim}
\hat{\text{IKL}}(\sigma_0) = n^{-1}\sum_{i=1}^n \bigg[ \frac{\hat{\tilde{v}}({\bm x}_i)}{v({\bm x}_i)} + \log\frac{v({\bm x}_i)}{\hat{\tilde{v}}({\bm x}_i)} \bigg]^\zeta, 
\end{equation}
which is an estimate of the Integrated Kullback\textendash Leibler (IKL)
divergence, that is
\[
\text{IKL}(\sigma_{0})=\int\text{KL}\big[\text{N}\{\mu({\bm x}),\tilde{v}({\bm x})\}, \text{N}\{\mu({\bm x}),v({\bm x})\}\big]^{\zeta}p({\bm x})d{\bm x}\propto\int\bigg\{\frac{\tilde{v}({\bm x})}{v({\bm x})}+\log\frac{v({\bm x})}{\tilde{v}({\bm x})}\bigg\}^{\zeta}p({\bm x})d{\bm x},
\]
where $v({\bm x})={\bf x}\ts{{\bf V}}{\bf x}$ and $\tilde{v}({\bm x})={\bf x}\ts\tilde{{\bf V}}{\bf x}$ are the posterior variances of $\mu(\bm x)$ under the two alternative covariance matrices for $\bm \beta$. $\zeta$ is a positive constant and $\text{N}(\cdot, \cdot)$ indicates the normal distribution. In (\ref{eq:IKLestim}) $\tilde{{\bf V}}$ is replaced by $\hat{\tilde{{\bf V}}} = (\bm{\mathcal I}\hat{\bm{\Sigma}}_{\nabla}^{-1}\bm{\mathcal I}+{\bf S}^{\bm{\gamma}})^{-1}$, where $\hat{\bm{\Sigma}}_{\nabla}$ is the regularised estimator proposed in SM \ref{sec:regSand}. 
%Notice that $\mu({\bm x})$, ${{\bf V}}$ and $\tilde{{\bf V}}$ implicitly depend on $\sigma_{0}$. 
Objective (\ref{eq:IKLestim}) is deterministic and one dimensional, hence it can be efficiently minimised using standard root-finding methods, such as bisection. In our experience, the objective is generally smooth and it has a unique minimum. Decreasing (increasing) $\sigma_0$ leads to wigglier (smoother) fits and increases (decreases) $\tilde{v}({\bf x})/v({\bf x})$.

Our approach is motivated as follows.  Notice that the Gibbs posterior can be seen as a posterior based on the misspecified parametric likelihood, formed by the ELF density. \cite{muller2013risk} proves that, while the posterior of misspecified models is asymptotically Gaussian with mean vector $\hat{\bm{\beta}}$ and covariance matrix ${{\bf V}}$, this posterior is asymptotically worse, in
terms of frequentist risk, than a posterior having `sandwich' covariance $\tilde{{\bf V}}$. Given that credible intervals can be derived within a decision-theoretic
framework by adopting an appropriate loss function (see for instance \cite{robert2007bayesian}, Section 5.5.3), and that M{\"u}ller's work considers general losses, it is clear that the intervals based on $\tilde{{\bf V}}$ should have better asymptotic frequentist properties. Hence, we minimise (\ref{eq:IKLestim}) w.r.t. $\sigma_0$ so that the marginal posterior distribution of $\mu({\bm x})$, which is based on ${{\bf V}}$, is as close as possible to that based on $\tilde{{\bf V}}$. We choose $\zeta=1/2$, and in general we suggest setting $0<\zeta<1$, to make $\text{IKL}$ more robust to the occasional large discrepancies between $\tilde{v}({\bm x})$ and $v({\bm x})$, which can occur where the design points are sparse. Notice that the KL divergence is asymmetric, hence IKL is not invariant to the ordering of $v({\bm x})$ and $\tilde{v}({\bm x})$. We prefer the ordering used here, because it penalises under-coverage ($\tilde{v}({\bm x})>v({\bm x})$) more than over-coverage ($\tilde{v}({\bm x})<v({\bm x})$).

Recall that we are minimising the discrepancy between the marginal posteriors for $\mu({\bf x})$ based on ${{\bf V}}$ and $\tilde{{\bf V}}$ because the latter offers better asymptotic frequentist properties. However, \cite{muller2013risk} clarifies that adopting $\tilde{{\bf V}}$ does not lead to a posterior achieving the lowest possible asymptotic risk. Hence, it is reasonable to expect that intervals based on the true marginal variance of $\hat{\mu}({\bm x})={\bf x}\ts\hat{\bm{\beta}}$ under $\mathbb{P}$ would offer better coverage, especially in small samples. To provide such an alternative to the sandwich estimator, SM \ref{app:bootCal} proposes a bootstrapping procedure for estimating a different $\text{IKL}$ loss, where an estimate of $\text{var}\{\hat{\mu}({\bm x})\}$ under $\mathbb{P}$ substitutes $\tilde{v}({\bf x})$. 

Motivated by the non-parametric spline-based context considered here, we have chosen to explicitly calibrate the posterior of $\mu({\bm x})$, rather than that of ${\bm \beta}$. However, the following argument suggests that the calibration procedure proposed here should lead to approximately calibrated intervals for ${\bm \beta}$ too. Assume that ${\bf X}$ is a $n \times d$ full rank matrix  and suppose that minimising the IKL loss leads to a value of $\sigma_0$ such that ${\bf x}_i\ts{\bf V} {\bf x}_i = {\bf x}_i\ts\tilde{\bf V}{\bf x}_i$, for $i=1,\dots,n$. Then the properties of Kronecker products lead to ${\bf X} \otimes_r {\bf X} \, \text{vec}({\bf V}) = {\bf X} \otimes_r {\bf X}\, \text{vec}(\tilde{\bf V})$, where $\otimes_r$ indicates the row-wise Kronecker product, such that the $i$-th row of ${\bf X} \otimes_r {\bf X}$ is ${\bf x}_i \otimes {\bf x}_i$. Given that ${\bf X}$ is of full rank $d$, then the symmetry of ${\bf V}$ and the fact that $\text{rank}({\bf X} \otimes_r {\bf X}) \geq d(d+1)/2$ imply that ${\bf V} = \tilde{\bf V}$. 

\cite{muller2013risk} proves that $\tilde{{\bf V}}$ achieves a lower asymptotic frequentist risk than ${\bf V}$ in a setting where the prior is increasingly dominated by the likelihood as $n$ increases, so that the asymptotic variance of $\hat{\bm \beta}$ does not depend on the prior. In a penalised cubic regression spline context such dominance occurs when the spline basis dimension $d=O(n^\alpha)$ for $\alpha<1/5$. This includes the regime considered by \cite{kauermann2009some}, when demonstrating the statistical validity of GAM inference based on Laplace approximate marginal likelihood smoothing parameter estimation. However, other regimes are also possible \citep[e.g.][]{claeskens2009asymptotic}, and the question of relative risk is then open. % Candidate15775

\section{Simulated examples} \label{sec:simulExamples}

Before applying the proposed quantile regression framework to load forecasting, we test it on two simulated examples. In particular, in Section \ref{sec:addExample} we fit an additive quantile model to homoscedastic data with $\sigma({\bm x}) = \sigma_0$, while in Section \ref{sec:hetero} we consider heteroscedastic data, where adequate interval coverage can be achieved only by letting the learning rate and the ELF loss bandwidth vary with $\bm x$. 

\subsection{An additive example} \label{sec:addExample}

Consider the following additive model
\begin{equation} \label{eq:addModel}
y_i = x_i + x_i^2 - z_i + 2\text{sin}(z_i) + 0.1v_i^3 + 3\text{cos}(v_i) + e_i,
\end{equation}
where $e_i \sim \text{gamma}(3, 1)$, $x_i \sim \text{unif}(-4, 4)$, $z_i \sim \text{unif}(-8, 8)$ and $v_i \sim \text{unif}(-4, 4)$. We aim at estimating the conditional quantiles corresponding to $\tau = 0.01, 0.05, 0.5, 0.95$ and $0.99$. Hence, we fit an additive quantile regression model for each $\tau$, using the ELF loss. We determine the loss bandwidth as in Section \ref{sec:approxLapl}, where the $\mathbb{E}(y|\bm x) = \alpha(\bm x)$ and $\text{var}(y) = \kappa^2$ are estimated using a Gaussian GAM. Fitting this model has a negligible impact on the computational cost, as it has to be done only once, before calibrating $\sigma_0$. We select $\sigma_0$ either by minimising LAML w.r.t. both $\sigma_0$ and $\bm \gamma$, or by the calibration method of Section \ref{sec:tuningLearn}. We consider two versions of the latter, one based on the sandwich covariance matrix $\tilde{\bf V}$, the other on the bootstrapping routine of SM \ref{app:bootCal}. We also include quantile regression by gradient boosting, as implemented in the \verb|mboost| R package \citep{hothorn2010model}.

We simulate 100 datasets from (\ref{eq:addModel}), using either $n=10^3$ or $n=10^4$, and we fit an additive model for each $\tau$ using each approach. The fitted model includes a smooth effect for each covariate, based on cubic regression splines bases of rank 30. 
% Cubic regression spline for qgam and p-splines for mboost
The boosting approach requires also selecting the degrees of freedom of each effect, which we set to 6. The number of boosting iterations was selected by minimising the out-of-bag empirical risk, based on the pinball loss and on 100 bootstrap datasets. The boosting step size was equal to $0.1$ when $n=10^3$ and $1$ when $n=10^4$. To select $\sigma_0$ by posterior calibration, we minimised the estimated ${\text{IKL}}$ loss using Brent's method \citep{brent2013algorithms}. The bootstrap version of the procedure was based on 100 bootstrap samples. 
%The following results seem to be robust to the choice of $\epsilon$, at least in the range $0.01$ to $0.05$. 

%
\begin{table}
\centering
\begin{tabular}{rlllll}
  \hline
 $\tau$ & 0.01 & 0.05 & 0.5 & 0.95 & 0.99 \\ 
  \hline
CAL Boot & \textbf{0.273}(0.04) & \textbf{0.237}(0.03) & 0.309(0.04) & 0.722(0.1) & 1.104(0.23) \\ 
  CAL Sand & 0.274(0.04) & \textbf{0.237}(0.03) & \textbf{0.303}(0.04) & \textbf{0.717}(0.1) & \textbf{1.097}(0.22) \\ 
  LAML & 0.284(0.03) & 0.249(0.03) & 0.307(0.04) & 0.926(0.14) & 1.284(0.16) \\ 
  BOOST & 0.369(0.08) & 0.272(0.04) & 0.321(0.05) & 0.814(0.11) & 1.674(0.38) \\ 
  \hline
  CAL Boot & 0.102(0.01) & 0.093(0.01) & 0.125(0.01) & 0.314(0.03) & 0.543(0.08) \\ 
  CAL Sand & \textbf{0.1}(0.01) & \textbf{0.092}(0.01) & \textbf{0.123}(0.01) & \textbf{0.307}(0.03) & \textbf{0.535}(0.08) \\ 
  LAML & 0.113(0.01) & 0.1(0.01) & 0.126(0.01) & 0.406(0.04) & 0.897(0.09) \\ 
  BOOST & 0.107(0.01) & 0.094(0.01) & \textbf{0.123}(0.01) & \textbf{0.307}(0.04) & 0.561(0.08) \\ 
   \hline
\end{tabular} 
\caption{Additive example: mean(std. dev.) of the RMSEs between true and estimated quantiles, for each quantile and method for $n=10^3$ (top rows) and $n=10^4$ (bottom rows). The lowest RMSE(s), for each sample size and quantile, is \textbf{bold}.}
\label{tab:addExLoss} 
\end{table} 

Table \ref{tab:addExLoss} reports the average RMSE ($[n^{-1}\sum_i \{\hat{\mu}({\bm x}_i) - \mu^0({\bm x}_i)\}^2]^{1/2}$). When $n=10^3$, the RMSEs achieved by the two calibration approaches are strictly lower than those achieved by boosting. The performance of boosting is closer to that of our method when $n=10^4$, which suggests that, for fixed model complexity, the advantage of using a smooth loss is inversely proportional to the amount of data available. LAML selection of $\sigma_0$ leads to worse results relative to our method, especially for the highest quantiles. Hence, in this example the calibration procedure based on $\tilde{\bf V}$ leads to quantile estimates that are as accurate as those produced by bootstrap based calibration, and more accurate than those obtained by boosting, and much cheaper to compute. On an Intel 2.50GHz CPU, calibrating $\sigma_0$ using $\tilde{\bf V}$ takes around $1.4$s for $n=10^3$ and $11$s for $n=10^4$, when $\tau = 0.5$. Under bootstrapping, the calibration takes $13$s and $126$s, while selecting the number of boosting steps takes around $134$s (2000 steps) and $150$s (550 steps). However, selecting the number of boosting step takes much longer for $\tau=0.01$: $0.6$h ($3\times10^4$ steps) and $0.5$h ($6000$ steps). For the same quantile $\tilde{\bf V}$-based calibration takes $3.5$s and $20$s, while the bootstrap version takes $45$s and $350$s. In practice \verb|mboost|'s computing times are longer, as the cross-validation needs to run beyond the optimal step size, which is not known in advance. 

\begin{figure} 
\centering
\includegraphics[scale=0.44]{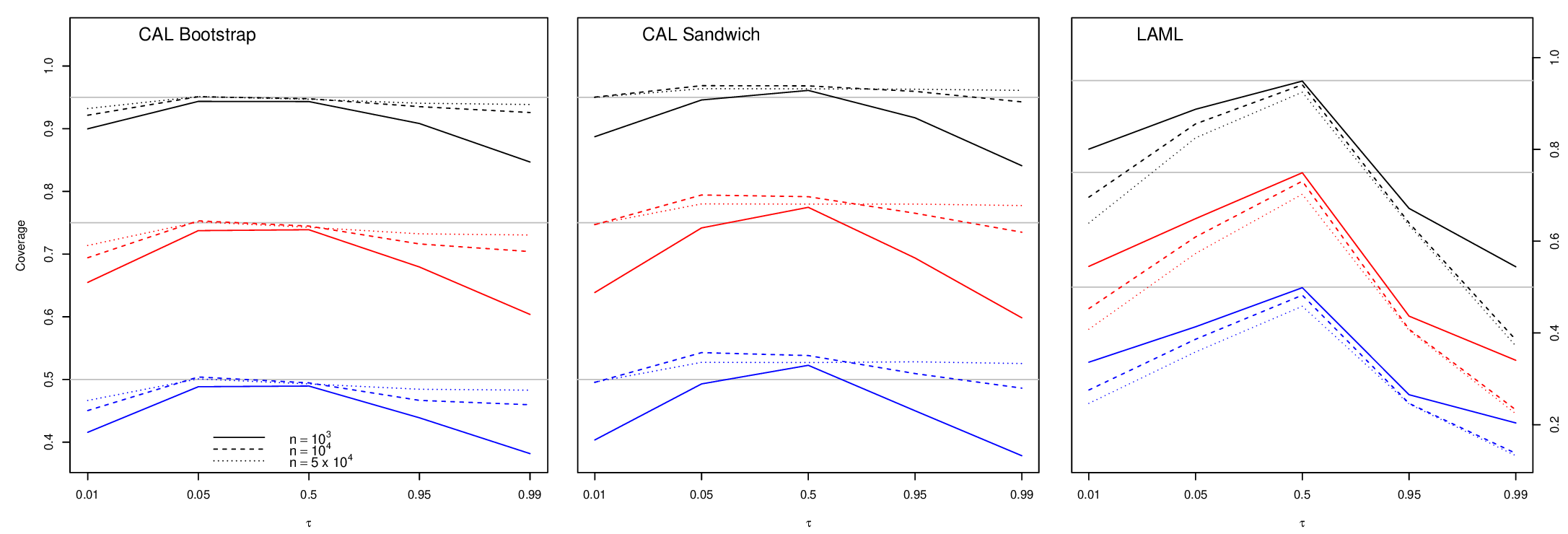} 
\caption{Additive example: empirical coverage achieved by selecting $\sigma_0$ by calibration (using either the sandwich covariance $\tilde{\bf V}$ or bootstrapping) or by LAML minimisation, for each $\tau$ and sample size $n$. The horizontal lines indicate the nominal coverage levels.}
\label{fig:cal3D}
\end{figure}

Figure \ref{fig:cal3D} shows the empirical coverage, at $95$, $75$ and $50\%$ level, achieved by the credible intervals for $\mu({\bm x})$, using calibration or LAML to select $\sigma_0$. %The coverage was calculated using 100 simulations from model (\ref{eq:addModel}) with $n=10^3$, $10^4$ and $5 \times 10^4$.
We do not check the coverage achieved by gradient boosting, because analytic formulas are unavailable and confidence intervals must be obtained by bootstrapping with each bootstrap replicate as computationally expensive as the original fit. Notice that the coverage achieved using LAML for selecting $\sigma_0$ is well below nominal levels for most quantiles, and does not improve as $n$ increases. Instead, under bootstrap or sandwich-based calibration, coverage improves with $n$. In particular, the bootstrap version attains close to nominal coverage for most quantiles when $n=10^4$, and almost perfect coverage for $n=5\times 10^4$. The sandwich version shows a similar improvement but, for large $n$, it selects larger values of $\sigma_0$ than the bootstrap version, which leads to slightly wider intervals. The bootstrap-based calibration selects a lower value of $\sigma_0$ because it minimises a version of the IKL criterion which takes into account the bias of the fit (see SM \ref{app:bootCal} for details), which is directly proportional to $\sigma_0$.

\subsection{An heteroscedastic example} \label{sec:hetero}

Here we consider the following heteroscedastic data generating process
\begin{equation} \label{eq:heteroModel}
y_i \sim \text{SkewNorm}\{\xi(x_i), \omega(x_i), \theta\}, \;\;\; \xi(x_i) = x_i + x_i^2, \;\;\;  \omega(x_i) = 1.5 + \text{sin}(2x_i), \;\;\; \theta = 4,
\end{equation}
where $\xi$, $\omega$, and $\theta$ are the location, scale and shape parameters of the Skew-Normal distribution \citep{azzalini1985class}, while $x_i \sim \text{unif}(-4, 4)$. We simulate $n = 2000$ data points from (\ref{eq:heteroModel}) and we fit quantile models for the median and the 95th percentile. In particular, we consider a simplified model where $\sigma({x})=\sigma_0$, and a full model where the learning rate and the loss bandwidth vary with $x$. As explained in Section \ref{sec:approxLapl}, this requires fitting a location-scale model to estimate the conditional mean and variance of $y$. We use a Gaussian GAM where both the mean and the variance of $y$ depend on $x$, which we fit using the methods of \cite{wood2016smoothing}. We model the quantiles using cubic regression spline bases of rank 30, and we adopt the same basis for the mean and variance of the Gaussian GAM.  

\begin{figure} 
\centering
\includegraphics[scale=0.53]{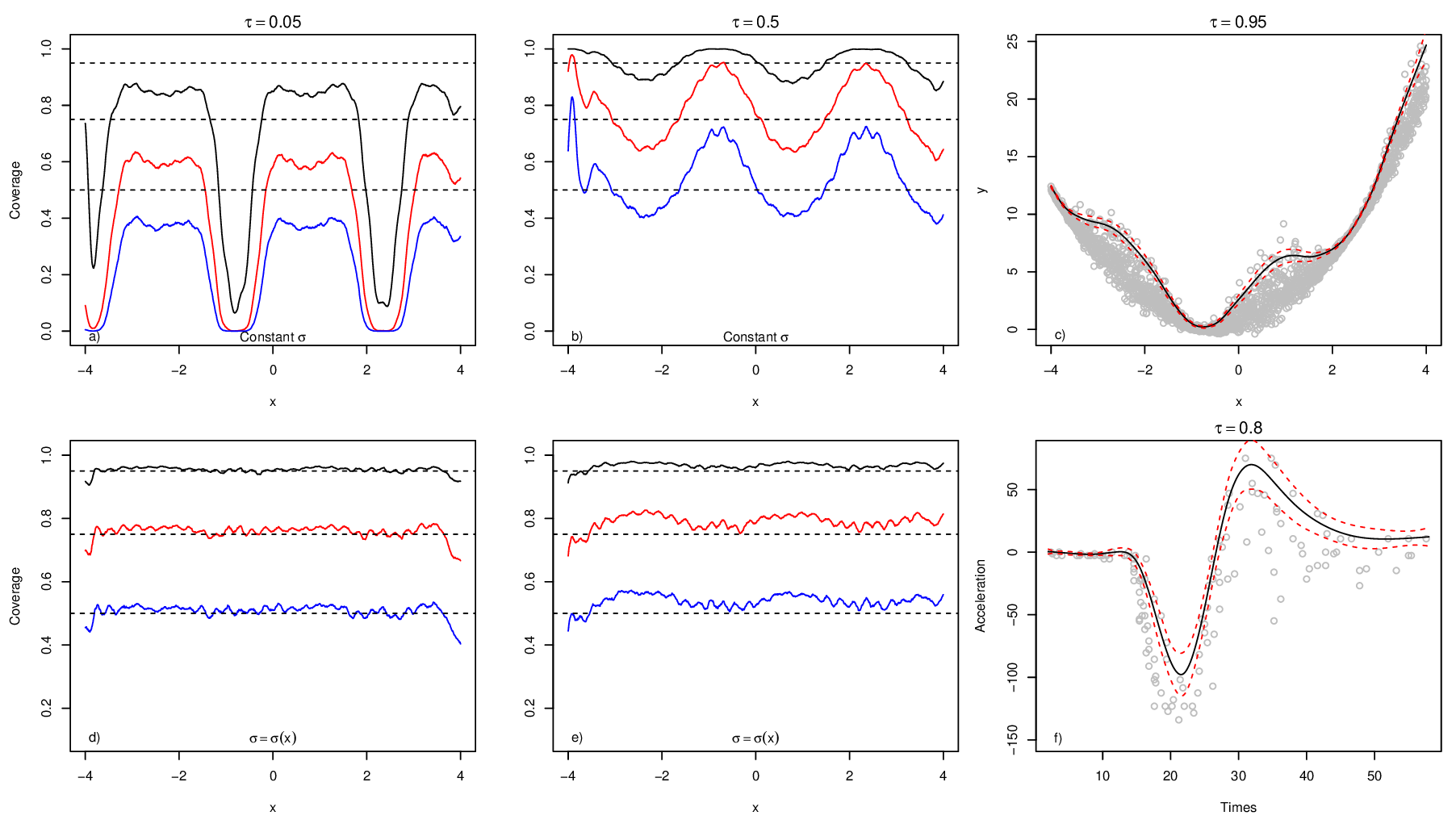} 
\caption{$a$, $b$, $d$ and $e$: nominal (dashed) vs empirical (solid) coverage at $50$, $75$ and $95\%$ level, using a simplified ($a$ and $b$) or full quantile model ($d$ and $e$). $c$: full fit using $\tau = 0.95$ with data from model (\ref{eq:heteroModel}). $f$: full fit for quantile $\tau = 0.8$ using the motorcycle dataset.}
\label{fig:heteroPlots}
\end{figure}

The first two columns in Figure \ref{fig:heteroPlots} compare nominal and empirical coverage of credible intervals for $\mu(x)$, obtained by fitting 5000 datasets simulated from (\ref{eq:heteroModel}) using $\tilde{\bf V}$-based calibration. Clearly, the simplified model provides unreliable intervals even at the median, while the intervals of the full model are much closer to nominal levels. Figure \ref{fig:heteroPlots} also shows a fit for quantile $\tau = 0.8$ of the motorcycle dataset \citep{silverman1985some}. This was obtained using an adaptive P-spline basis of rank 20 to model the quantile curve and the mean of the Gaussian GAM, and a thin-plate spline basis of rank 10 for the conditional variance.

\section{Probabilistic load forecasting} \label{sec:gefcom}

GAMs have proved highly successful at EDF, because they can capture the complex relations existing between electricity load and several meteorological, economic and social factors, while retaining a high degree of interpretability, which is critically important during exceptional events, when manual intervention might be required. However,  the cost structure relevant to an electrical utility implies that only certain conditional quantile estimates are of high operational interest. This, and the difficulty of finding a parametric model for the load distribution that holds at several levels of aggregation, makes of semi-parametric quantile regression an attractive alternative to traditional GAMs.

In this section we consider the three datasets shown in Figure \ref{fig:loadScores}. The first is the dataset used in the load forecasting track of the Global Energy Competition 2014 (GEFCom2014). This covers the period between January 2005 and December 2011, and it includes half-hourly load consumption and temperatures. 
%The latter were measured at 25 weather stations, but here we average the temperature records of only four stations to obtain a single variable. See \cite{gaillard2016additive} for details on how this subset of stations was selected. 
The other two datasets contain half-hourly electricity demand from the UK and French grids. The first covers the period between January 2011 and June 2016, the second between January 2013 and December 2017. We integrate them with hourly temperature data from the National Centers for Environmental Information (NCEI) and M{\'e}t{\'e}o France. We aim at predicting 20 conditional quantiles, equally spaced between $\tau = 0.05$ and $\tau = 0.95$. Given that load consumption is strongly dependent on the time of the day, it is common practice \cite[e.g.][]{gaillard2016additive} to fit a different model for each half-hour. To limit the computational burden, here we consider only the period between between 11:30 and 12am. We use the period 2005-09 of the GEFCom2014 data for training, the last two years for testing. Similarly, we test each method on the last 24 and 12 months of, respectively, the UK and the French data set.

\begin{figure} 
\centering
\includegraphics[scale=0.6]{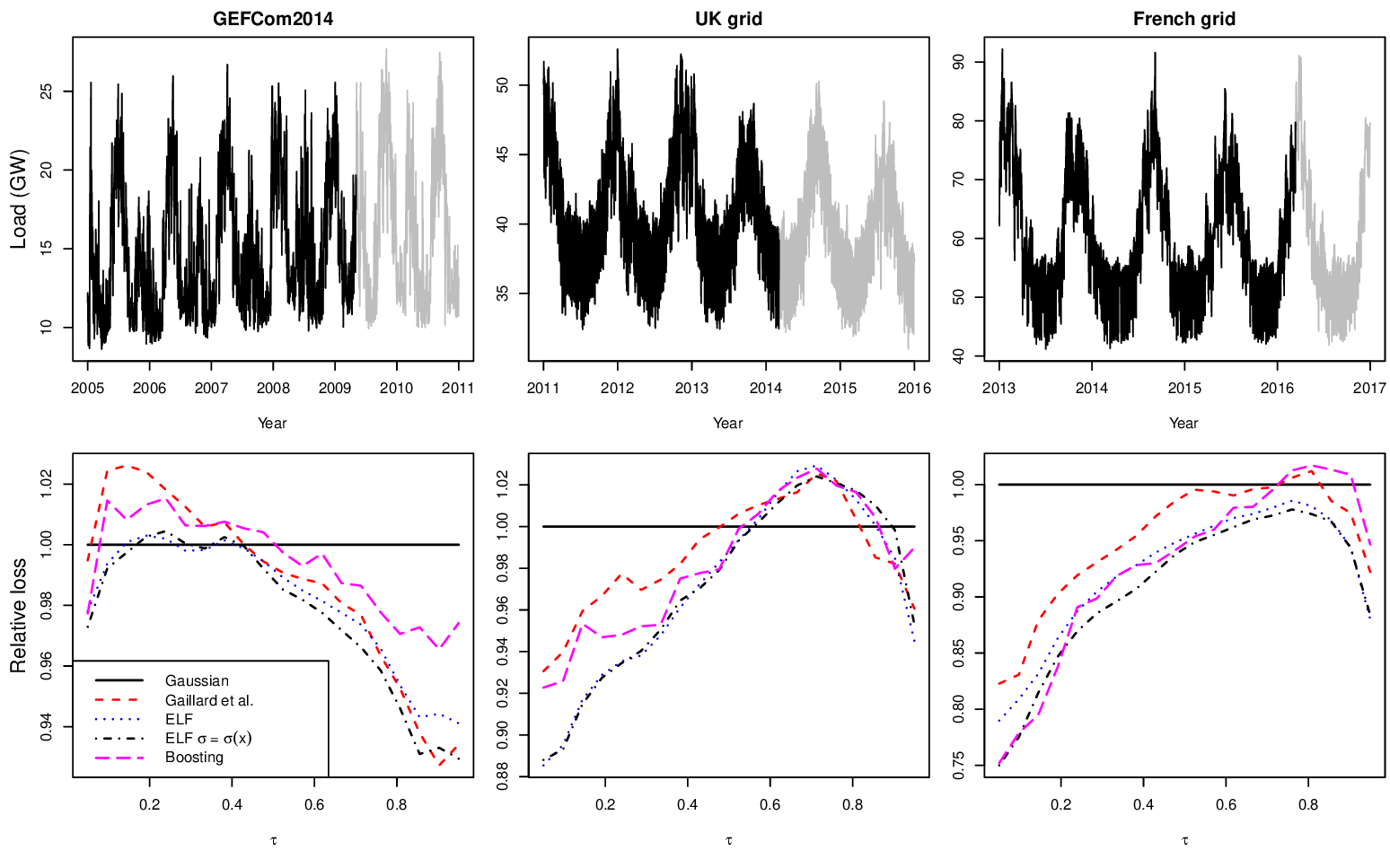} 
\caption{Top row: daily electricity loads, observed between 11:30 and 12am, from the GEFCom2014 challenge, the UK and the French grid, divided between training (black) and testing (grey) sets. Bottom row: relative pinball losses for each $\tau$ and method.}
\label{fig:loadScores}
\end{figure}

\cite{gaillard2016additive} proposed a quantile regression method which ranked 1st on both the load and the price forecasting track of GEFCom2014. This is a two-step procedure, which was partially motivated by the lack of reliable software for fitting additive quantile models. Very briefly, their method firstly fits a Gaussian additive model to model mean load and a second one to model the variance of the residuals from the first fit. Then, for each quantile, they fit a linear quantile regression to model the load, using the effects estimated by the Gaussian fits as covariates. We compare their method to our proposal and to gradient boosting, using the set of covariates proposed by \cite{gaillard2016additive}: hourly temperatures ($T_t$); smoothed temperature ($T^s_t$), obtained using $T_{t}^s = \alpha T_{t-1}^s + (1-\alpha)T_t$ with $\alpha = 0.95$; a cyclic variable indicating the position within the year ($S_t$); a factor variable indicating the day of the week ($D_t$); a sequential index representing time ($t$); the observed load at the same time of the previous day ($L_{t-48}$). Hence, the model for quantile $\tau$ is
$$
q_\tau(L_t) = \psi_{D_t} + f_1^{15,6}(T_t) + f_2^{15,6}(T^s_t) + f_3^{{20},8}(S_t) + f_4^{{10},4}(L_{t-48}) + f_5^{4,4}(t),
$$
where all smooth effects are based on cubic regression splines and, for instance, $f_1^{15,6}$ indicates that we used a basis of rank 15 and 6 degrees of freedom (the latter need to be chosen only under boosting). $\psi_{D_t}$ is a fixed effect, whose value depends on $D_t$. 

To let the learning rate and the ELF loss smoothness depend on the covariates, and to implement the method of \cite{gaillard2016additive}, we need also a variance model for the preliminary location-scale Gaussian GAM fit. We use 
$$
\log{\text{var}(\text{Load}_t)} = \tilde{\psi}_{D_t} + \tilde{f}_1^{10}(T^s_t) + \tilde{f}_2^{10}(S_t), 
$$
where the degrees of freedom do not need to be chosen. We consider two version of ELF-based quantile regression, one where $\sigma({\bm x}) = \sigma_0$ and another where $\sigma$ depends on the covariates, using the location-scale model just described. In both cases, we calibrate $\sigma_0$ using $\tilde{\bf V}$-based calibration, because for this application it gives exactly the same results as bootstrapping, but at a much lower computational cost. For boosting, we use 100 bootstrap replicates to select the number of steps, while the step-size is equal to 0.1. Having tuned $\sigma_0$ and the number of boosting steps on the training sets, we forecast electricity load one week ahead, on rolling basis, using the test sets. See SM \ref{app:electrDetails} for further details. 

The bottom plots in Figure \ref{fig:loadScores} show, for each $\tau$ and dataset, the pinball losses incurred on the testing sets, divided by the pinball loss of a Gaussian additive fit. Notice that both versions of ELF-based quantile regression do better than a Gaussian fit for most  quantiles and that they are more accurate than the alternative quantile regression methods on each data set. Remarkably, on the GEFCom2014 data set, the proposed approach is superior to that of \cite{gaillard2016additive}, which was developed in the context of that challenge. Letting the learning rate $\sigma$ depend on $\bm x$ leads to slightly improved performance on the GEFCom2014 and French data, but not on the UK data. On this data all quantile regression methods struggle to beat the Gaussian GAM around $\tau = 0.7$, which might be attributable to temperature anomaly that characterised the winter 2015/16 in the UK.

% Figure \ref{fig:GEFscores} also shows, for each quantile, the effects of $T_t^s$ and $S_t$, estimated using our method. It is interesting to notice that the temperature effect is more pronounced on high quantiles. Instead, the demand drop around year-end affects mostly low quantiles. These effects lead to changes in scale and shape of the conditional load distribution, which cannot be captured by a Gaussian GAM, and might explain its poor performance on the UK dataset.
 
Regarding computing time, $\tilde{\bf V}$-based calibration takes around 2s on the GEFCom2014 training data. This includes the time needed to fit a Gaussian GAM, when $\sigma(\bm x) = \sigma_0$. Fitting a location-scale Gaussian GAM takes around 20s, but this needs to be done only once, and the output is used for all quantiles. For gradient boosting, the number of steps which minimises the cross-validated risk criterion varies widely across quantiles. In fact, for $\tau\approx0.95$ more than $5400$ steps are needed, while for $\tau\approx 0.4$ it is sufficient to use $550$ steps. Using \verb|mboost|, cross-validation takes around 800s in the first case and 80s in the second. In practice, the optimal number of steps is not known in advance, hence it is necessary to cross-validate beyond the optimal point. We considered up to $1\times10^4$ steps for all quantiles, which translates to roughly $1450$s. In the simulation setting considered here, where the smoothing parameters and regression coefficients are updated every week, our method has the further advantage that it is possible to initialise using the latest fit.

%\begin{figure} 
%\centering
%\includegraphics[scale=0.55]{GEFeffect.eps} 
%\caption{Smooth effects for the median, estimated using calibrated quantile regression with the full GEFCom2014 dataset.}
%\label{fig:GEFeffect}
%\end{figure}

\section{Conclusion} \label{sec:conclusion}

Motivated by the need for more flexible GAM models at EDF, this work provides a computationally stable and efficient framework for fitting additive quantile regression models. The learning rate and all smoothing parameters are selected automatically and calibrated uncertainty estimates are provided at no additional computational cost. This was achieved by coupling the general Bayesian belief updating framework of \cite{bissiri2016general} with a statistically advantageous smooth generalisation of the `pinball' loss and a novel efficient calibration procedure based on a Bayesian sandwich covariance estimator. Adopting the smooth ELF loss was key to computationally efficiency, because it allowed us to exploit the fast stable method of \cite{wood2016smoothing}, when selecting the smoothing parameter by marginal loss minimisation. While working with a smooth loss is not novel in itself \citep[see e.g.][]{oh2012fast, yue2011bayesian}, the key contribution here is having selected the loss smoothness so as to minimise the asymptotic MSE of the estimated regression coefficients.

In practical terms the electricity load forecasting example demonstrates the practical utility of the proposed approach. Performance is better than that of gradient boosting, and at much lower computational cost (even more so when uncertainty estimates are required). Similarly, the methods are competitive with the ad hoc approach of \cite{gaillard2016additive}, on the very example that method was designed for.

\bibliographystyle{chicago}
\bibliography{biblio.bib}

\newpage

\begin{center}
{\large\bf Supplementary material to \\ ``Fast calibrated additive quantile regression''}
\end{center}

\renewcommand{\appendixpagename}{}
\begin{appendices}

\renewcommand{\theequation}{S\arabic{equation}}

\setcounter{equation}{0}

\section{Details regarding the ELF loss} \label{app:densDetails}

\subsection{The Extended Log-F (ELF) density} \label{sec:exLogDens}

This section explains how the new ELF density relates to the log-F density of \cite{jones2008class}. Consider the family of densities with exponential tails described
by \cite{jones2008class}
\[
p_{G}(y)=K_{G}^{-1}(\alpha,\beta)\exp\big\{\alpha y-(\alpha+\beta)G^{[2]}(y)\big\},
\]
where $\alpha,\beta>0$, $K_{G}(\alpha,\beta)$ is a normalising constant,
$
G^{[2]}(y)=\int_{-\infty}^{y}\int_{-\infty}^{t}g(z)dzdt=\int_{-\infty}^{y}G(t)dt,
$
while $g(z)$ and $G(z)$ are, respectively, the p.d.f and c.d.f. of a (fictitious) r.v. $z$.  Importantly, this family nests the AL distribution, which is recovered by choosing $g(z)$ to be the Dirac delta and by imposing $\alpha=1-\tau$, $\beta=\tau$, with $0<\tau<1$. Adding location and scale parameters is trivial. 

We substitute the Dirac delta with a smoother p.d.f.. This is achieved by choosing $G(z)=G(z|\lambda)=\Phi(z|0,\lambda)=\exp(z/\lambda)/\{1+\exp(z/\lambda)\}$,
which is the c.d.f. of a logistic random variable centered at zero
and with scale $\lambda$. Notice that, as $\lambda \rightarrow 0$,
we have that $\Phi(z|\lambda)\rightarrow \mathbbm{1}(z>0)$ which is the c.d.f. corresponding to the Dirac delta density. With this choice we have
$
\Phi^{[2]}(y|\lambda)=\lambda\log\{1+\exp(y/\lambda)\}
$,
which leads to
\begin{equation}
p_{F}(y)=\frac{e^{(1-\tau) y}(1+e^{\frac{y}{\lambda}})^{-\lambda}}{\lambda \text{Beta}\big[\lambda(1-\tau),\lambda\tau\big]}.\label{eq:plainDens}
\end{equation}
where $\text{Beta}(\cdot, \cdot)$ is the beta function. The location-scale extension of (\ref{eq:plainDens}) is simply
\begin{equation}
\tilde{p}_{F}(y)=\frac{1}{\sigma}p_{F}\big\{(y-\mu)/\sigma \big\}=\frac{e^{(1-\tau)\frac{y-\mu}{\sigma}}(1+e^{\frac{y-\mu}{\lambda\sigma}})^{-\lambda}}{\lambda\sigma\text{Beta}\big[\lambda(1-\tau),\lambda\tau\big]},
\end{equation}
Imposing $\lambda = 1$ leads to the log-F density of \cite{jones2008class}. Sections \ref{sec:ELFderiv} and \ref{sec:ELFsatlik} contain additional details regarding the new density. Most of these are necessary to fit semi-parametric additive models using the methods described in Sections \ref{sec:fixedSigma} and \ref{sec:tuningLearn}.

\subsection{Derivation of optimal loss smoothness} \label{app:OptSmoothLoss}

Before deriving the AMSE under the ELF loss, we need to put forward some definitions and to prove some preliminary results. Our proofs follow closely those of \cite{kaplan2017smoothed}, but there are some differences in the type of smooth loss we adopt and in the assumptions we make, hence we can not simply refer to their results.

Let ${\bm w}_i = \nabla_{\bm \beta} \tilde{\rho}\{{y}_i-\mu({\bm x}_i)\}$ be the gradient of the $i$-th component of the ELF loss. Also, define $h = \lambda \sigma$, $u = y - \mu(\bm x)$, its p.d.f $f_{u|{\bm x}}(u|{\bm x})$, c.d.f. $F_{u|{\bm x}}(u|{\bm x})$ and make the following assumptions:
\begin{enumerate}[label=\emph{\alph*})]
\item \label{iidX} the pairs $\{{\bm x}_i, {y}_i\}$, with $i=1,\dots,n$, are i.i.d.; 
\item \label{quantileAss} $\mathbb{P}(u_{i} < 0 | {\bm x}_i) = \tau$ for almost all $\bm x \in \mathcal{X}$, the support of $\bm x$, and $i \in \{1, \dots, n\}$;
\item \label{fboundAss} for almost all ${\bm x} \in \mathcal{X}$ and $u$ in a neighborhood of zero, $f_{u|{\bm x}}(u|{\bm x})$ is three times continuously differentiable and there exists a bounded function $C({\bm x})$ such that $|f_{u | {\bm x}}^{(s)}( u | {\bm x})| \leq C({\bm x})$ for $s \in \{1,2,3\}$,  and $\mathbb{E}\left\{ C({\bm x}) \|{\bf x}\|^2 \right\} < \infty$;
\item \label{hrateAss} $h = O(n^{-\psi})$ with $1/5 < \psi < 2/5$;
\item \label{identif} ${\bm \beta}_0$ is the unique solution of $\mathbb{E}[{\bf x}_i\{\mathbb{I}(y_i-{\bf x}\ts\bm \beta)>0) - 1 + \tau\}]={\bf 0}$ on ${\bm \beta} \in \mathcal{B}$, the parameter space;
\item \label{nonSing} ${\bm \Sigma}_f = \mathbb{E}\{{\bf x}_i{\bf x}_i \ts f_{u|{\bm x}}(0|{\bm x}_i)\}$ and $\mathbb{E}({\bf x}_i{\bf x}_i\ts)$ are non-singular.
\end{enumerate}
We will also need the following facts:
\begin{enumerate}[label=\emph{\alph*})]\setcounter{enumi}{6}
\item \label{logiSym} the logistic p.d.f. $\phi$ is symmetric around zero;
\item \label{logiProp} $\int_{-\infty}^{+\infty} |v^4 \phi(v)| dv < \infty$, $\int_{-\infty}^{+\infty} v^2\Phi(v)\phi(v) dv < \infty$ and $\int_{-\infty}^{+\infty} v \Phi(v) \phi(v) dv = 1/2$.
\end{enumerate}

In a regression context, equating to zero the first derivative of the ELF loss w.r.t. $\bm \beta$ leads to 
\begin{equation} \label{eq:kernSmo}
{\bm m}_n({\bm \beta}) = \frac{1}{\sqrt n}\sum_{i=1}^n \bm w_i({\bm \beta}) = \frac{1}{\sqrt n} \sum_{i=1}^n \bigg ( {\bf x}_i \bigg[ \Phi\bigg \{\frac{y_i-{\bf x}\ts{\bm \beta}}{h} \bigg\} - 1 + \tau \bigg] \bigg) = {\bf 0}.
\end{equation}
We start by deriving asymptotic expressions for $\mathbb{E}({\bm w}_i)$ and $\mathbb{E}({\bm w}_i{\bm w}_i\ts)$.

\begin{Lemma}
Under assumptions \ref{quantileAss} and \ref{fboundAss}, and for every $i$, we have that

\begin{equation}
\mathbb{E}({\bm w}_{i}) = \frac{1}{6}h^2\pi^2\mathbb{E}\left\{ f'_{u | {\bm x}}(0 | {\bm x}_{i}){\bf x}_{i}\right\} + O(h^4).
\label{E(Wi)}
\end{equation}
\end{Lemma}

\begin{proof}

\begin{align*}
\mathbb{E}({\bm w}_i|{\bm x}_i) =& \; {\bf x}_{i} \left[ \int_{-\infty}^{+\infty} \left\{\Phi\left(\frac{u}{h}\right)-1+\tau\right\}\mathrm{d}F_{u | {\bm x}}(u | {\bm x}_{i})\right]  \\
=& \left[ {\bf x}_{i}\left\{\Phi\left(\frac{u}{h}\right)-1+\tau\right\}F_{u | {\bm x}}(u | {\bm x}_{i}) \right]_{-\infty}^{+\infty} - \frac{{\bf x}_{i}}{h}\left\{ \int_{-\infty}^{+\infty}F_{u | {\bm x}}(u | {\bm x}_{i})\phi\left(\frac{u}{h}\right)\mathrm{d}u\right\}  \\
\text{\footnotesize{using \ref{quantileAss}} $\rightarrow$ } =& \; {\bf x}_{i}\tau - {\bf x}_{i}\left\{ \int_{-\infty}^{+\infty}F_{u | {\bm x}}(hv | {\bm x}_{i})\phi\left(v\right)\mathrm{d}v\right\} \\
\text{\footnotesize{using \ref{logiSym}} $\rightarrow$ } =& \; {\bf x}_{i}\tau - {\bf x}_{i}F_{u | {\bm x}}(0 | {\bm x}_{i}) - {\bf x}_{i}\frac{h^2}{2}f'_{u | {\bm x}}(0 | {\bm x}_{i})\int_{-\infty}^{+\infty}v^2 \phi(v)\mathrm{d}v   \\
& - {\bf x}_{i}\frac{h^4}{24}\int_{-\infty}^{+\infty}v^4f'''_{u | {\bm x}}(\tilde{h}v | {\bm x}_{i})\phi(v)\mathrm{d}v \\
\text{\footnotesize{using \ref{quantileAss}} $\rightarrow$ } =& \; {\bf x}_{i}\frac{h^2\pi^2}{6}f'_{u | {\bm x}}(0 | {\bm x}_{i}) - {\bf x}_{i}\frac{h^4}{24}\int_{-\infty}^{+\infty}v^4f'''_{u | {\bm x}}(\tilde{h}v | {\bm x}_{i})\phi(v)\mathrm{d}v,
\end{align*}
where $\tilde{h} = \tilde{h}(v) : \mathbb R \rightarrow [0,h]$.  Then taking expectation w.r.t. ${\bm x}$, leads to 
\[ \mathbb{E}({\bm w}_{i}) =  \frac{1}{6}h^2\pi^2\mathbb{E}\left\{f'_{u | {\bm x}}(0 | {\bm x}_{i}){\bf x}_{i} \right\} + O(h^4),   \] 
where we used \ref{fboundAss}, \ref{logiProp} and Jensen's inequality to bound the remainder, that is
\begin{align*}
\left\| \mathbb{E}\left\{ {\bf x}_{i}\int_{-\infty}^{+\infty}v^4f'''_{u | {\bm x}}(\tilde{h}v | {\bm x}_{i})\phi(v)\mathrm{d}v  \right\}\right\| & \leq \mathbb{E}\left\{ C({\bm x}_{i})\|{\bf x}_{i}\| \int_{-\infty}^{+\infty}|v^4 \phi(v)| \mathrm{d}v \right\} = O(1). \qedhere
\end{align*}

\end{proof}

\begin{Lemma}
Under assumptions \ref{quantileAss} and \ref{fboundAss}, and for every $i$, we have that
\begin{equation}
\mathbb{E}({\bm w}_{i}{\bm w}_{i}\ts) = \tau(1-\tau)\mathbb{E}({\bf x}_{i}{\bf x}_{i}\ts)-h\mathbb{E}\left\{ f_{u | {\bm x}}(0 | {\bm x}_{i}){\bf x}_{i}{\bf x}_{i}\ts\right\} + O(h^2).
\label{E(Wi'Wi)}
\end{equation}
\end{Lemma}

\begin{proof}

We have that
\[\mathbb{E}({\bm w}_{i}{\bm w}_{i}\ts) = \mathbb{E}\left[ {\bf x}_{i}{\bf x}_{i}\ts \int_{-\infty}^{+\infty}\left\{\Phi\left(\frac{u}{h}\right)-1+\tau\right\}^{2}\mathrm{d}F_{u | {\bm x}}(u | {\bm x}_{i})  \right], \]
where
\begin{align*}
\int_{-\infty}^{+\infty}\left\{\Phi\left(\frac{u}{h}\right)-1+\tau\right\}^{2}\mathrm{d}F_{u | {\bm x}}(u | {\bm x}_{i}) =& \left[\left\{\Phi\left(\frac{u}{h}\right)-1+\tau\right\}^{2}F_{u | {\bm x}}(u | {\bm x}_{i}) \right]_{-\infty}^{+\infty} \\ 
& - \frac{2}{h}\int_{-\infty}^{+\infty}F_{u | {\bm x}}(u | {\bm x}_{i})\left\{\Phi\left(\frac{u}{h}\right)-1+\tau\right\}\phi\left(\frac{u}{h}\right)\mathrm{d}u\\
\text{\footnotesize{using \ref{quantileAss}} $\rightarrow$ } =& \; \tau^2 - 2 \int_{-\infty}^{+\infty}F_{u | {\bm x}}(hv | {\bm x}_{i})\left\{\Phi\left(v\right)-1+\tau\right\}\phi\left(v\right)\mathrm{d}v \\
\text{\footnotesize{using \ref{fboundAss}} $\rightarrow$ } =& \; \tau^2 - 2 \tau\int_{-\infty}^{+\infty}\left\{\Phi\left(v\right)-1+\tau\right\}\phi\left(v\right)\mathrm{d}v \\ 
& - 2hf_{u | {\bm x}}(0 | {\bm x}_{i})\int_{-\infty}^{+\infty}v\left\{\Phi\left(v\right)-1+\tau\right\}\phi\left(v\right)\mathrm{d}v \\
& - h^2 \int_{-\infty}^{+\infty}v^2f'_{u | {\bm x}_{i}}(\tilde{h}v | {\bm x}_{i})\left\{\Phi\left(v\right)-1+\tau\right\}\phi\left(v\right)\mathrm{d}v.
\end{align*}
Then we have
\begin{align*}
\int_{-\infty}^{+\infty}\left\{\Phi\left(v\right)-1+\tau\right\}\phi\left(v\right)\mathrm{d}v =& \left[ \frac{\left\{\Phi\left(v\right)-1+\tau\right\}^2}{2} \right]^{+\infty}_{-\infty} = \frac{2\tau-1}{2}, 
\end{align*}
and, using \ref{logiSym} and \ref{logiProp}, we have
$$
\int_{-\infty}^{+\infty}v\left\{\Phi\left(v\right)-1+\tau\right\}\phi\left(v\right)\mathrm{d}v = \int_{-\infty}^{+\infty}v\Phi\left(v\right)\phi\left(v\right)\mathrm{d}v + (-1+\tau)\int_{-\infty}^{+\infty}v\phi\left(v\right)\mathrm{d}v = \frac{1}{2},
$$
thus we obtain
\begin{align*}
\int_{-\infty}^{+\infty}\left\{\Phi\left(\frac{u}{h}\right)-1+\tau\right\}^{2}\mathrm{d}F_{u | {\bm x}}(u | {\bm x}_{i}) =& \; \tau^2 - \tau(2\tau - 1) - hf_{u | {\bm x}}(0 | {\bm x}_{i}) \\ 
& - h^2 \int_{-\infty}^{+\infty}v^2f'_{u | {\bm x}_{i}}(\tilde{h}v | {\bm x}_{i})\left\{\Phi\left(v\right)-1+\tau\right\}\phi\left(v\right)\mathrm{d}v.
\end{align*}
Using \ref{fboundAss} and \ref{logiProp} we obtain
\begin{align*}
& \left\| \mathbb{E}\left[ {\bf x}_{i}{\bf x}_{i}\ts\int_{-\infty}^{+\infty}v^2f'_{u | {\bm x}_{i}}(\tilde{h}v | {\bm x}_{i})\left\{\Phi\left(v\right)-1+\tau\right\}\phi(v)\mathrm{d}v  \right]\right\| \\
& \leq \mathbb{E}\left[C({\bm x}_{i})\|{\bf x}_{i}\|^2 \int_{-\infty}^{+\infty} |v^2 \left\{\Phi\left(v\right)-1+\tau\right\} \phi(v)| \mathrm{d}v \right]= O(1),
\end{align*}
which leads to
\[ \mathbb{E}({\bm w}_{i}{\bm w}_{i}\ts) = \tau(1-\tau)\mathbb{E}({\bf x}_{i}{\bf x}_{i}\ts)-h\mathbb{E}\{f_{u | {\bm x}}(0 | {\bm x}_{i}){\bf x}_{i}{\bf x}_{i}\ts\} + O(h^2). \qedhere \]

\end{proof}

Under assumptions \ref{iidX} to \ref{nonSing}, part of Lemma 9 in \cite{kaplan2017smoothed} proves that
\begin{equation}
\sqrt{n}(\hat{\bm{\beta}}-\bm{\beta}_{0})=-\bigg[\frac{1}{\sqrt{n}}\nabla_{\bm{\beta}}^{T}{\bm m}_{n}(\bm{\beta})\Big|_{\bm \beta = \bm \beta_0}\bigg]^{-1}{\bm m}_{n}(\bm{\beta}_{0}) + O_p\left(\frac{1}{\sqrt{n}}\right),
\end{equation}
and 
\begin{equation}
\mathbb{E}\bigg[\frac{1}{\sqrt{n}}\nabla_{\bm{\beta}}^{T}{\bm m}_{n}(\bm{\beta})\Big|_{\bm \beta = \bm \beta_0}\bigg] = \mathbb{E}\left[{\bf x}_i {\bf x}_i\ts f_{u|{\bm x}}(0|{\bm x}_i) \right] + O(h^2),
\end{equation}
where ${\bm m}_n({\bm \beta})$ has been defined in (\ref{eq:kernSmo}). 
Define $\bm H = n^{-1/2}\nabla_{\bm{\beta}}^{T}{\bm m}_{n}(\bm{\beta}) = n^{-1} \sum_{i=1}^n h^{-1}\phi(u_i/h){\bf x}_i{\bf x}_i\ts$ evaluated at $\bm \beta = \bm \beta_0$, and ${\bm \Sigma}_f = \mathbb{E}\left[{\bf x}_i {\bf x}_i\ts f_{u|{\bm x}}(0|{\bm x}_i) \right]$. Under the assumptions adopted so far, we have that
\begin{align*} 
\text{var}(H_{jk}) \, = \, & \frac{1}{nh^2}\text{var}\left\{({\bf x}_i{\bf x}_i\ts)_{jk}\phi\left(\frac{u_i}{h}\right)\right\} \leq \frac{1}{nh^2}\mathbb{E}\left\{({\bf x}_i{\bf x}_i\ts)_{jk}^2\phi\left(\frac{u_i}{h}\right)^2\right\} \\ \, = \, & \frac{1}{nh^2} \mathbb{E}\left\{({\bf x}_i {\bf x}_i\ts)_{jk}^2 \int \phi\left(\frac{u}{h}\right)^2 f_{u|{\bm x}}(u|{\bm x}_i)du\right\} = \frac{1}{nh} \mathbb{E}\left\{({\bf x}_i {\bf x}_i\ts)_{jk}^2 \int \phi\left(v\right)^2 f_{u|{\bm x}}(hv|{\bm x}_i)dv\right\} \\ = &  \frac{1}{nh} \mathbb{E}\left[({\bf x}_i {\bf x}_i\ts)_{jk}^2 \left\{ f_{u|{\bm x}}(0|{\bm x}_i) \int \phi\left(v\right)^2 dv  + hf'_{u|{\bm x}}(0|{\bm x}_i)\int v \phi\left(v\right)^2 dv + O(h^2) \right\} \right] \\ = &  O\left(\frac{1}{nh}\right) + O\left(\frac{1}{n}\right),
\end{align*} 
so that we can write $\bm H = \mathbb{E}(\bm H) + {\bm C}$, where $\bm C$ is a matrix such that with $\mathbb{E}(\bm C) = {\bm 0}$, and with elements of size $O(1/\sqrt{nh})$. Now 
\begin{align*} 
\mathbb{E}({\bm H}^{-1}) = & \mathbb{E}[\{\mathbb{E}({\bm H}) + {\bm C}\}^{-1}] = \mathbb{E}\{\mathbb{E}({\bm H})^{-1} - \mathbb{E}({\bm H})^{-1}{\bm C}\mathbb{E}({\bm H})^{-1} + O(||{\bm C}||^2)\} \\  = & {\bm \Sigma}_f^{-1} + O(h^2) + O(1/nh),
\end{align*}
so we have
\begin{align*} 
\text{AMSE}(h) = & \mathbb{E}\{n(\hat{\bm \beta}-{\bm \beta}_0)(\hat{\bm \beta}-{\bm \beta}_0)\ts\} \nonumber \\ = & 
{\bm \Sigma}_f^{-1}{\mathbb{E}({\bm m}_n{\bm m}_n\ts)}{\bm \Sigma}_f^{-1} + R(h)  \\ = & {\bm \Sigma}_f^{-1}{\mathbb{E}\left(\frac{1}{n}\sum_{i=1}^n{\bm w}_i\sum_{i=1}^n{\bm w}_i\ts\right)}{\bm \Sigma}_f^{-1} + R(h) \\ = &
\frac{1}{n}{\bm \Sigma}_f^{-1}{\mathbb{E}\left(\sum_{i=1}^n{\bm w}_i{\bm w}_i\ts + \sum_{i=1}^n \sum_{j\neq i}{\bm w}_i{\bm w}_j\ts\right)}{\bm \Sigma}_f^{-1} + R(h) \\ \text{\footnotesize{using \ref{iidX}} $\rightarrow$ } = &
{\bm \Sigma}_f^{-1}\left\{\mathbb{E}\left({\bm w}_i{\bm w}_i\ts\right) + (n-1)\mathbb{E}({\bm w}_i)\mathbb{E}({\bm w}_j)\ts\right\}{\bm \Sigma}_f^{-1} + R(h) \\ \text{\footnotesize{using (\ref{E(Wi'Wi)})} $\rightarrow$ } = &
{\bm \Sigma}_f^{-1}\left[\tau(1-\tau)\mathbb{E}({\bf x}_{i}{\bf x}_{i}\ts)-h\mathbb{E}\big\{ f_{u | {\bm x}}(0 | {\bm x}_{i}){\bf x}_{i}{\bf x}_{i}\ts\right\} + O(h^2) \\ + & (n-1)\mathbb{E}({\bm w}_i)\mathbb{E}({\bm w}_j)\ts\big]{\bm \Sigma}_f^{-1}  + R(h) \\  = &
{\bm \Sigma}_f^{-1}{\bf V}^{\frac{1}{2}}\left\{{\bf I}-h {\bf A}{\bf A}\ts  + (n-1){\bf V}^{-\frac{1}{2}}\mathbb{E}({\bm w}_i)\mathbb{E}({\bm w}_j)\ts{\bf V}^{-\frac{1}{2}}\right\}{\bf V}^{\frac{1}{2}}{\bm \Sigma}_f^{-1}  + R(h),
\end{align*}
where $R(h) = O(h^2) + O(1/nh)$ and ${\bf A} = \mathbb{E}\{f_{u|{\bm x}}(0|{\bm x}_i)^{1/2}{\bf V}^{-1/2}{\bf x}_i\}$. Using (\ref{E(Wi)}) we have that
\begin{align*}
(n-1){\bf V}^{-\frac{1}{2}}\mathbb{E}({\bm w}_{i})\mathbb{E}({\bm w}_{i})\ts{\bf V}^{-\frac{1}{2}} & = (n-1)\left[ \left\{ h^2{\bf B}+O(h^4)\right\}\left\{ h^2{\bf B}\ts+O(h^4)\right\} \right] \\
& = (n-1)\left\{h^4{\bf B}{\bf B}\ts + O(h^6)\right\} \\
& = nh^4{\bf B}{\bf B}\ts + O(h^4) + O(nh^6).
\end{align*}
where ${\bf B} = \pi^2\mathbb{E}\{ f'_{u | {\bm x}}(0 | {\bm x}_{i}){\bf V}^{-1/2}{\bf x}_{i}\}/6$, so that 
$$
\text{AMSE}(h) = {\bm \Sigma}_f^{-1}{\bf V}^{\frac{1}{2}}\left\{{\bf I}-h {\bf A}{\bf A}\ts + nh^4{\bf B}{\bf B}\ts\right\}{\bf V}^{\frac{1}{2}}{\bm \Sigma}_f^{-1}  + O(h^2) + O(1/nh) + O(nh^6).
$$
We minimize the first term on r.h.s. of the expression for the AMSE, while discarding the remaining terms. This is justified as long as first term dominates the rest, which happens for $h = O(n^{-\psi})$ with $1/5 < \psi < 2/5$. Notice that in \cite{kaplan2017smoothed} the $O(1/nh)$ term above appears to be $O(1/\sqrt{nh})$, which could be discarded only if $\psi < 1/3$. This is an important difference as, under the ELF loss, it would invalidate the optimal $h = O(n^{-1/3})$ rate derived below. 

Proving that 
$$
h^* = \left(\frac{{\bf A}\ts{\bf A}}{4n{\bf B}\ts{\bf B}}\right)^{\frac{1}{3}}.
$$
minimises $\text{tr}(nh^4{\bf B}{\bf B}\ts-h {\bf A}{\bf A}\ts)$ is straightforward. Under the further assumption that the distribution of $u_i$ does not depend on $\bm x_i$, we have that
$$
{\bf A}\ts{\bf A} = \frac{f_{u}(0)\text{tr}({\bf V}{\bf V}^{-1})}{\tau(1-\tau)} = \frac{f_{u}(0)d}{\tau(1-\tau)},
$$
and 
$$
{\bf B}\ts {\bf B} = \frac{1}{36}\pi^4f'_{u}(0)^2\mathbb{E}\left({\bf x}_{i}\right)\ts{\bf V}^{-1}\mathbb{E}\left({\bf x}_{i}\right) = \frac{1}{\tau(1-\tau)36}\pi^4f'_{u}(0)^2,
$$
where the second equality in the last equation is proved by \cite{kaplan2017smoothed}, under the reasonable assumption that one of the elements of ${\bf x}_i$ is fixed to a non-zero real number (i.e. the model contains an intercept). Hence, we have
$$
\tilde{h}^* = \left[\frac{d}{n}\frac{9f_{u}(0)}{\pi^4f'_{u}(0)^2}\right]^{\frac{1}{3}}.
$$
which completes the proof.  

\subsection{Derivatives of the ELF log-likelihood} \label{sec:ELFderiv}

The logarithm of the ELF density is 
\[
\text{ll}(y)=\log\tilde{p}_{F}(y - \mu) = (1-\tau)\frac{y-\mu}{\sigma}-\lambda\log\bigg(1+e^{\frac{y-\mu}{\lambda\sigma}}\bigg)-\log\bigg[\lambda\sigma\text{Beta}\big\{\lambda(1-\tau),\lambda\tau \big\}\bigg],
\]
When evaluating it numerically, it is important to approximate $\log(1+e^{z})$ with $z+e^{-z}$ when $z = (y-\mu)/\lambda\sigma >18$, as suggested by \cite{machler2012accurately}. 
The gradient is
\[
\frac{\partial \text{ll}(y)}{\partial\mu}=\frac{1}{\sigma}\bigg\{\Phi(y|\mu,\lambda\sigma)-1+\tau\bigg\},
\;\;\;
\frac{\partial \text{ll}(y)}{\partial\sigma}=\frac{y-\mu}{\sigma^{2}}\bigg\{\Phi(y|\mu,\lambda\sigma)-1+\tau\bigg\}-\frac{1}{\sigma},
\]
where $\Phi(y|\mu,\lambda\sigma)$ is the logistic c.d.f., with location $\mu$ and scale $\lambda\sigma$. The Hessian is
\[
\frac{\partial^{2}\text{ll}(y)}{\partial\mu^{2}}=-\frac{1}{\sigma}\phi(y|\mu,\lambda\sigma),
\]
\[
\frac{\partial^{2}\text{ll}(y)}{\partial\sigma^{2}} = 2\frac{y-\mu}{\sigma^{3}}\Bigg\{1-\tau-\Phi(y|\mu,\lambda\sigma)-\frac{1}{2}(y-\mu)\phi(y|\mu,\lambda\sigma)\Bigg\}+\frac{1}{\sigma^{2}},
\]
\[
\frac{\partial^{2}\text{ll}(y)}{\partial\mu\partial\sigma}=-\frac{1}{\sigma^{2}}\Bigg\{(y-\mu)\phi(y|\mu,\lambda\sigma)+\Phi(y|\mu,\lambda\sigma)-1+\tau\Bigg\},
\]
where $\phi(y|\mu,\lambda\sigma)$ is the logistic p.d.f.. Define $z=(y-\mu)/(\lambda\sigma)$ so that 
$
\Phi(y|\mu,\lambda\sigma)=\Phi(z|0,1)=\Phi(z)=(1+e^{-z})^{-1},
$
is the sigmoid function. Also, define $\Phi^{(k)}(z)=\partial\Phi^{(k)}(z)/\partial z^{k}$. Then, derivatives of higher order are
\[
\frac{\partial^{3}\text{ll}(y)}{\partial\mu^{3}}=\frac{\Phi^{(2)}(z)}{\lambda^{2}\sigma^{3}},
\;\;\;\;\;\;
\frac{\partial^{4}\text{ll}(y)}{\partial\mu^{4}}=-\frac{\Phi^{(3)}(z)}{\lambda^{3}\sigma^{4}},
\]
\[
\frac{\partial^{3}\text{ll}(y)}{\partial\sigma^{3}}=-\frac{3}{\sigma}\frac{\partial^{2}\text{ll}(y)}{\partial\sigma^{2}}+\frac{\lambda z^{2}}{\sigma^{3}}\Bigg\{3\Phi^{(1)}(z)+z\Phi^{(2)}(z)+\frac{1}{\lambda z^{2}}\Bigg\},
\]
\[
\frac{\partial^{4}\text{ll}(y)}{\partial\sigma^{4}}=-\frac{4}{\sigma}\bigg\{2\frac{\partial^{3}\text{ll}(y)}{\partial\sigma^{3}}+\frac{3}{\sigma}\frac{\partial^{2}\text{ll}(y)}{\partial\sigma^{2}}\bigg\}-\frac{\lambda z^{3}}{\sigma^{4}}\Bigg\{4\Phi^{(2)}(z)+z\Phi^{(3)}(z)-\frac{2}{\lambda z^{3}}\Bigg\},
\]
\[
\frac{\partial^{3}\text{ll}(y)}{\partial\mu^{2}\partial\sigma}=\frac{1}{\lambda\sigma^{3}}\big\{z\Phi^{(2)}(z)+2\Phi^{(1)}(z)\big\}, \;\;\;\;\;\;
\frac{\partial^{4}\text{ll}(y)}{\partial\mu^{3}\partial\sigma}=-\frac{1}{\lambda^{2}\sigma^{4}}\big\{z\Phi^{(3)}(z)+3\Phi^{(2)}(z)\big\},
\]
\[
\frac{\partial^{3}\text{ll}(y)}{\partial\mu\partial\sigma^{2}}=\frac{1}{\sigma^{3}}\bigg\{2\big\{\Phi(z)-1+\tau\big\}+4z\Phi^{(1)}(z)+z^{2}\Phi^{(2)}(z)\bigg\},
\]
\[
\frac{\partial^{4}\text{ll}(y)}{\partial\mu\partial\sigma^{3}}=-\frac{3}{\sigma}\frac{\partial^{3}\text{ll}(y)}{\partial\mu\partial\sigma^{2}}-\frac{z}{\sigma^{4}}\big\{6\Phi^{(1)}(z)+6z\Phi^{(2)}(z)+z^{2}\Phi^{(3)}(z)\big\},
\]
\[
\frac{\partial^{4}\text{ll}(y)}{\partial\mu^{2}\partial\sigma^{2}}=-\frac{1}{\lambda\sigma^{4}}\bigg\{ z^{2}\Phi^{(3)}(z)+6z\Phi^{(2)}(z)+6\Phi^{(1)}(z)\bigg\},
\]
where
$
\Phi^{(1)}(z)=\Phi(z)\big\{1-\Phi(z)\big\},
$
$
\Phi^{(2)}(z)=\Phi^{(1)}(z)-2\Phi^{(1)}(z)\Phi(z),
$
and
$
\Phi^{(3)}(z)=\Phi^{(2)}(z)-2\Phi^{(2)}(z)\Phi(z)-2\Phi^{(1)}(z)^{2}.
$

\subsection{ELF saturated log-likelihood and deviance} \label{sec:ELFsatlik}

To find the saturated log-likelihood, $\text{ll}_{s}$, we need to maximise $\tilde{p}_{F}(y - \mu)$ w.r.t. $\mu$. This leads to 
\[
\hat{\mu}=\lambda\sigma\log\Big(\frac{\tau}{1-\tau}\Big)+y.
\]
so the saturated log-likelihood is
\[
\text{ll}_{s}(y)=(1-\tau)\lambda\log\big(1-\tau\big)+\lambda\tau\log(\tau)-\log\bigg[\lambda\sigma\text{Beta}\big\{\lambda(1-\tau),\lambda\tau\big\}\bigg],
\]
and has derivatives
$
\partial \text{ll}_{s}(y)/\partial\sigma=-{\sigma}^{-1},
$
$
\partial^{2}\text{ll}_{s}(y)/\partial\sigma^{2}=\sigma^{-2}.
$
The saturated loss we refer to in the main text is simply $\tilde{\text{ll}} = - (1-\tau)\lambda\log\big(1-\tau\big)-\lambda\tau\log(\tau)$.
The deviance is
\[
\text{Dev}(y)=2\big[\text{ll}_{s}(y)-\text{ll}(y)\big]=2\Bigg[(1-\tau)\lambda\log\big(1-\tau\big)+\lambda\tau\log(\tau)-(1-\tau)\frac{y-\mu}{\sigma}+\lambda\log\bigg\{1+e^{\frac{y-\mu}{\lambda\sigma}}\bigg\}\Bigg],
\]
which is identical to the loss-based deviance definition $2[\text{lo}(\mu, \sigma)-\tilde{\text{ll}}]$ in the main text.

\section{Stabilising computation under the ELF density } \label{app:stableNewDensity}

\subsection{Dealing with zero weights in PIRLS} \label{sec:zeros1}

Quantile regression with the ELF loss requires that we work with many weights that can be very close to zero, while the corresponding log-likelihood or deviance derivative is far from zero. This can lead to a situation in which the vector containing $w_i z_i$ is well scaled, while the vector containing $\sqrt{|w_i|}z_i$ is very poorly scaled. This scaling problem can reverse the usual stability improvement of QR-based least squares estimation over direct normal equation solution.

We adopt the notation of \cite{wood2011fast}. Let ${\bar {\bf W}}$ be a diagonal matrix with ${\bar { W}}_{ii} = |w_i|$ and let $\bf E$ be a matrix such that ${\bf S}^{\bm \gamma} = {\bf E}\ts{\bf E}$. Then let $\bm{\mathcal{\bm Q}}\bm{\mathcal{\bm R}}$ be the QR decomposition of $\sqrt{\bar {\bf W}}{\bf X}$ and define the further QR decomposition
$$
 \begin{pmatrix}
  \bm{\mathcal{\bm R}}\\
  {\bf E}
 \end{pmatrix}
= {\bf Q}{\bf R}. 
$$
Define the matrix ${\bf Q}_1 = \bm{\mathcal{\bm Q}}{\bf Q}[1{:}d, :]$,  where $d$ is the number of columns of $\bf X$ and ${\bf Q}[1{:}d, :]$ indicates the first $d$ rows of ${\bf Q}$. We also need to define the diagonal matrix ${\bf I}^-$, such that $I^-_{ii} $ is equal to 0 if $w_{i}>0$ and 1 otherwise, and the singular value decomposition ${\bf I}^-{\bf Q}_1 = \bf{U}\bf{D}\bf{V}\ts$. See \cite{wood2011fast} for details on how to deal with non-identifiable parameters.

Using this notation, \cite{wood2011fast} shows that
$$
\hat \bp = {\bf R}^{-1} {\bf V}({\bf I} - 2 {\bf D}^2)^{-1}{\bf V}\ts {\bf Q}_1\ts \sqrt{\bar {\bf W}} \bar {\bf z} = {\bf R}^{-1}  {\bf f},
$$
where $\bar {\bf z}$ is a vector such that $\bar {z}_i = {z}_i$ if $w_i \geq 0$ and $\bar {z}_i = -{z}_i$ otherwise, while the definition of $\bf f$ should be obvious. Now we can test for stability of the computation to the scaling of $\sqrt{\bar {\bf W}} \bar {\bf z}$ by testing whether
$$
{\bf R}{\bf Q}_1\ts \sqrt{\bar {\bf W}} \bar {\bf z} = {\bf X}\ts {\bf Wz},
$$
to sufficient accuracy. If it does not, then we recompute $\bf f$ using
$$
{\bf f} = {\bf V}({\bf I} - 2 {\bf D}^2){\bf V}\ts {\bf R}^{-1}{\bf X}\ts {\bf Wz}.
$$

If we define the matrices
$$
{\bf P} = {\bf R}^{-1} {\bf V}({\bf I} - 2 {\bf D}^2)^{-\frac{1}{2}},\;\;\;\;\;
{\bf K} = {\bf Q}_1 {\bf V}({\bf I} - 2 {\bf D}^2)^{-\frac{1}{2}},
$$
then another possibility, that may be more convenient when using $\hat \bp = {\bf PK}\ts \sqrt{\bar {\bf W}} \bar {\bf z}$, is to test whether $ {\bf K}\ts \sqrt{\bar {\bf W}} \bar {\bf z} = {\bf P}\ts {\bf Wz}$ holds to sufficient accuracy, and to use $\hat \bp = {\bf PP}\ts {\bf Wz}$ if not.

\subsection{Dealing with zero weights in LAML} \label{sec:zeros2}

Here we show how the gradient and Hessian of $\log |{\bf X}\ts {\bf WX} + {\bf S}_{\bm \lambda}|$, which are needed to maximise the LAML using Newton algorithm, can be computed in a stable manner. In order to be consistent with the notation of \cite{wood2011fast}, in this section we indicate the smoothing parameter vector with $\bm \lambda$, rather than with $\bm \gamma$, the penalty matrix with ${\bf S}_{\bm \lambda}$, rather than ${\bf S}^{\bm \gamma}$, and we define $\bm \rho = \log \bm \lambda$. Notice that $({\bf X}\ts {\bf WX} + {\bf S}_{\bm \lambda})^{-1} = {\bf P}{\bf P}\ts$, hence 
\begin{eqnarray*}
\pdif{\log |{\bf X}\ts {\bf WX} + {\bf S}_{\bm \lambda}|}{\rho_k} &=&
\text{tr} \left \{({\bf X}\ts {\bf WX} + {\bf S}_{\bm \lambda})^{-1} {\bf X}\ts\pdif{\bf W}{\rho_k} {\bf X} \right \} + \lambda_k \text{tr} \left \{
({\bf X}\ts {\bf WX} + {\bf S}_{\bm \lambda})^{-1} {\bf S}_k
\right \}\\ &=&
\text{tr} \left ({\bf P}\ts {\bf X}\ts \pdif{\bf W}{\rho_k} {\bf XP}  \right ) + \lambda_k
\text{tr} \left ({\bf P}\ts {\bf S}_k {\bf P} \right ).
\end{eqnarray*}
Then the Hessian is
\begin{eqnarray*}
\pddif{\log |{\bf X}\ts {\bf WX} + {\bf S}_{\bm \lambda}|}{\rho_k}{\rho_j} & = &
 \text{tr} \left \{({\bf X}\ts {\bf WX} + {\bf S}_{\bm \lambda})^{-1} {\bf X}\ts\pddif{\bf W}{\rho_k}{\rho_k} {\bf X} \right \} +
 \delta^j_k \lambda_j \text{tr} \left \{({\bf X}\ts {\bf WX} + {\bf S}_{\bm \lambda})^{-1} {\bf S}_j \right \} \\
 & - &
\text{tr} \left \{({\bf X}\ts {\bf WX} + {\bf S}_{\bm \lambda})^{-1}
\left ( {\bf X}\ts\pdif{\bf W}{\rho_k} {\bf X} + \lambda_j {\bf S}_j \right )
({\bf X}\ts {\bf WX} + {\bf S}_{\bm \lambda})^{-1}{\bf X}\ts\pdif{\bf W}{\rho_j} {\bf X} \right \}\\
& - &
\lambda_k \text{tr} \left \{({\bf X}\ts {\bf WX} + {\bf S}_{\bm \lambda})^{-1}
\left ( {\bf X}\ts\pdif{\bf W}{\rho_k} {\bf X} + \lambda_j {\bf S}_j \right )
({\bf X}\ts {\bf WX} + {\bf S}_{\bm \lambda})^{-1}{\bf S}_k \right \},
\end{eqnarray*}
so that
\begin{eqnarray*}
\pddif{\log |{\bf X}\ts {\bf WX} + {\bf S}_{\bm \lambda}|}{\rho_k}{\rho_j} & = &
\text{tr} \left ({\bf P}\ts {\bf X}\ts \pddif{\bf W}{\rho_k}{\rho_j} {\bf XP}  \right ) + \lambda_k \text{tr} \left ({\bf P}\ts {\bf S}_k {\bf P} \right ) \\
& - & \text{tr} \left ({\bf P}\ts {\bf X}\ts \pdif{\bf W}{\rho_j} {\bf XP} {\bf P}\ts {\bf X}\ts \pdif{\bf W}{\rho_k} {\bf XP} \right )\\
& - & \lambda_j \text{tr} \left ({\bf P}\ts {\bf S}_j {\bf P} {\bf P}\ts {\bf X}\ts \pdif{\bf W}{\rho_k} {\bf XP}  \right )
 -\lambda_k \text{tr} \left ( {\bf P}\ts {\bf X}\ts \pdif{\bf W}{\rho_j} {\bf XP}{\bf P}\ts {\bf S}_k {\bf P}\right )\\
& - & \lambda_j \lambda_k  \text{tr} ( {\bf P}\ts {\bf S}_j {\bf P}{\bf P}\ts {\bf S}_k {\bf P}).
\end{eqnarray*}
If we define the diagonal matrices ${\bf T}_j = \text{diag}(\ilpdif{w_i}{\rho_j})$ and
${\bf T}_{jk} = \text{diag}(\ilpddif{w_i}{\rho_j}{\rho_k})$, then this last expression corresponds to the equivalent formula in \cite{wood2011fast} and can be computed in the same way. The point of all this is that, if we followed the original formulation of \cite{wood2011fast}, we would be dividing by the (almost zero) weights in the definition of ${\bf T}_j$ and ${\bf T}_{jk}$. This is avoided here.

%For ML we need $\bar {\bf K} = {\bf XU}_1 \bar {\bf P}$ (but note that ${\bf XU}_1$ is what is supplied to {\tt gam.fit4}).

\section{Details regarding the calibration procedure}

\subsection{Calibration by bootstrapping} \label{app:bootCal}

Let ${\bf x}_i$ be the $i$-th vector of covariates and indicate with ${\bf X}$ the design matrix. Let $\mathbb{E}(z)$ and $\text{var}(z)$ be the expectation and variance, w.r.t. $\mathbb{P}$, of some r.v. z.  The aim here is estimating
\begin{equation} \label{eq:KLBoot}
\text{IKL}_{\mathbb{P}}(\sigma_{0})\propto\int\bigg\{\frac{\text{var}\{\hat{\mu}({\bm x})\}}{v({\bm x})}+\log\frac{v({\bm x})}{\text{var}\{\hat{\mu}({\bm x})\}} + \frac{1}{v({\bm x})}\big[{\mu}_0({\bm x}) - \mathbb{E}\{{\hat{\mu}({\bm x})\}}\big]^2\bigg\}^{\zeta}p({\bm x})d{\bm x},
\end{equation}
by bootstrapping (that is, sampling with replacement) the full dataset and then re-fitting the model on each bootstrap replicate. Relative to the IKL loss based on $\tilde{\bf V}$, notice that (\ref{eq:KLBoot}) contains also a term related to finite sample bias, which can be estimated at no extra cost using the same bootstrap samples used to estimate $\text{var}\{\hat{\mu}({\bm x})\}$. 

Indicate the $k$ bootstrap samples of $\bf y$ and ${\bf X}$ with ${\bf y}^1, \dots, {\bf y}^k$ and ${\bf X}_1, \dots,  {\bf X}_k$, respectively. Given these inputs, Algorithm \ref{calibration} gives the steps needed to estimate $\text{IKL}_{\mathbb{P}}(\sigma_{0})$, for fixed $\sigma_0$. An important feature of this procedure is that the smoothing parameters need to be estimated only once, using the full dataset, so that the cost of each bootstrap replicate is substantially less than the cost of a full model fit. Further, if the bootstrap samples are simulated only once, the marginal variance and bias estimates, and the resulting IKL loss, are deterministic functions of $\sigma_0$. 
\begin{algorithm}
%\setstretch{1.5}
\caption{Estimating $\text{IKL}_{\mathbb{P}}(\sigma_{0})$ for fixed $\sigma_0$}
\label{calibration}
Assume that $\tau$ is fixed and that $\lambda$ and $\sigma(\bm x)$ are functions of $\sigma_0$, determined explained as in Section \ref{sec:approxLapl}. Then the IKL loss is estimated as follows:
\begin{algorithmic}[1]
\STATE using the design matrix, $\bf X$, and response, $\bf y$, estimate $\bm \gamma$ by minimising (\ref{eq:LAMLsimple}). 
Given $\hat{\bm \gamma}$, estimate ${\bm \beta}$ by minimising the penalised loss (\ref{eq:DevCrit}) and obtain the reference estimate $\hat{\bm \mu}^0 = {\bf X} \hat{\bm \beta}$.
\STATE  For $j = 1,\dots,k$
\begin{enumerate}
\item Given $\hat{\bm \gamma}$, estimate $\bm \beta$ by minimising the penalised loss (\ref{eq:DevCrit}), based on the $j$-th bootstrap design matrix, ${\bf X}_j$, and response vector, ${\bf y}^j$. The resulting estimate is $\hat{\bm \beta}_j$.
\item Obtain the bootstrapped quantile prediction vector $\hat{\bm \mu}^j = {\bf X}\hat{\bm \beta}_j$.
\end{enumerate} 
\STATE Estimate the loss using 
$$
\hat{\text{IKL}}_{\mathbb{P}}(\sigma_0) = n^{-1}\sum_{i=1}^n \bigg[ \frac{\hat{\text{var}}\{\hat{\mu}({\bm x}_i)\}}{v({\bm x}_i)} + \log\frac{v({\bm x}_i)}{\hat{\text{var}}\{\hat{\mu}({\bm x}_i)\}}+ \frac{1}{v({\bm x}_i)}\big\{\hat{\mu}^0_i - \bar{\mu}({\bm x}_i)\big\}^2 \bigg]^\gamma. 
$$
where $\bar{\mu}({\bm x}_i)$ and $\hat{\text{var}}\{\hat{\mu}({\bm x}_i)\}$ are the sample mean and variance of $\hat{\mu}_i^1, \dots, \hat{\mu}_i^k$.
\end{algorithmic}
\end{algorithm}
\subsection{A regularised estimator for $\bm{\Sigma}_{\nabla}$} \label{sec:regSand} % Candidate6259 (whole section to appendix)

Let $\text{lo}=\text{lo}\{\mu(\bm x), \sigma(\bm x)\}$ be the ELF loss, then the covariance matrix of its gradient is
\begin{equation} \label{eq:sigNablaEgam}
{\bm \Sigma}_\nabla = \text{cov}\big(\nabla_{\bm{\beta}}\text{lo}|_{\bm \beta = \hat{\bm \beta}}\big)=\text{cov}({\bf x}\,\text{lo}')=\text{cov}\bigg({\bf x}\frac{1}{\sigma_0}\big[\Phi\{y|\hat{\mu}({\bm x}),\lambda\sigma_0\}-1+\tau\big]\bigg),
\end{equation}
where $\hat{\mu}({\bm x}) = {\bf x}\ts\hat{\bm \beta}$, $\text{lo}'=\partial \text{lo}/\partial\mu|_{\mu=\hat{\mu}}$ and $\Phi(y|a,b)$ is the logistic c.d.f. with mean $a$ and scale $b$. 
Without loss of generality, assume that $\sigma_0=1$ and $\tau>0.5$. Now define $s=\text{sign}(\text{lo}')$,
and $\omega=|\text{lo}'|$. The latter can be viewed as a weight
taking value in $[1-\tau,\tau]$. The covariance matrix could simply be
estimated by
\begin{equation} \label{eq:empCovEst}
\hat{\bm{\Sigma}}_{\nabla}=\frac{1}{n}\sum_{i=1}^n\omega_{i}^{2}{\bf x}_{i}{\bf x}_{i}\ts-{\bf x}_{\omega}{\bf x}_{\omega}\ts,\;\;\;\;\text{where}\;\;\;\;{\bf x}_{\omega}=\frac{1}{n}\sum_{i=1}^ns_i\omega_i{\bf x}_{i},
\end{equation}
but this estimator can be highly variable. In particular, set $\lambda\approx0$ and assume that $\hat{\mu}({\bm x})$ approximately divides the responses into $n\tau$ samples falling below it and $n(1-\tau)$ above it. Then, if $\tau\approx1$, $n\tau$ of the $\bf x$ vectors in (\ref{eq:empCovEst}) have weight $1-\tau\approx0$ and the remaining $n(1-\tau)$ have weight $\tau \approx 1$. Hence, when fitting extreme quantiles
with low loss smoothness (low $\lambda$), the estimator $\hat{\bm{\Sigma}}_{\nabla}$
will be based on very few observed ${\bf x}_{i}$'s, which is problematic
when $d=\text{dim}({\bf x})$ is close to $n(1-\tau)$. Obviously,
the same problem occurs when $\tau\approx0$.

We address this issue by regularising $\hat{\bm{\Sigma}}_{\nabla}$
using an inconsistent, but less variable, estimator. In particular, if we assume that $\text{lo}'$ and $(\text{lo}')^2$ are uncorrelated with, respectively, any element of ${\bf x}$ or of ${\bf x}{\bf x}\ts$, we have
\[
\text{cov}(\text{lo}'{\bf x})=\mathbb{E}\{(\text{lo}')^{2}\}\mathbb{E}({\bf x}{\bf x}\ts)-\mathbb{E}(\text{lo}')^{2}\mathbb{E}({\bf x})\mathbb{E}({\bf x})\ts,
\]
%Actually we don't need independence, but only $\text{cov}(x_{j},l^{\mu})=\text{cov}\{x_{j}^{2},(l^{\mu})^{2}\}=0$,
%where $x_{j}$ is the j-th element of vector $\bm{X}$. Look at Macros'
%answer to this question https://stats.stackexchange.com/questions/15978/variance-of-product-of-dependent-variables
which motivates the adoption of the estimator
\[
\tilde{\bm{\Sigma}}_{\nabla}=n^{-2}\bigg\{\Big(\sum_{i=1}^n\omega_{i}^{2}\Big){\bf X}\ts{\bf X}-\Big(\sum_{i=1}^ns_i\omega_i\Big)^{2}\bar{{\bf x}}\bar{{\bf x}}\ts\bigg\},
\]
where $\bar{{\bf x}}$ is the vector of column-means of ${\bf X}$. To see that $\tilde{\bm{\Sigma}}_{\nabla}$
is less variable than $\hat{\bm{\Sigma}}_{\nabla}$, consider a simplified
setting where $n\tau$ rows of ${\bf X}$ are randomly associated
with weight $1-\tau$, the rest with weight $\tau$ and assume, without
loss of generality, that we know that $\mathbb{E}({\bf x})=\bm{0}$.
Then we have that $\hat{\bm{\Sigma}}_{\nabla}\propto{\bf X}_{\tau}\ts{\bf X}_{\tau}+O\{(1-\tau)^{2}\}$
and $\tilde{\bm{\Sigma}}_{\nabla}\propto{\bf X}\ts{\bf X}$, where
${\bf X}_{\tau}$ is formed by the $n(1-\tau)$ rows of ${\bf X}$
associated with weight $\tau$. Under the assumption mentioned above both estimators
are consistent but, for $\tau\approx1$, $\hat{\bm{\Sigma}}_{\nabla}$
is effectively based on only $n(1-\tau)$ samples. Notice also that, because $\hat{\text{var}}(\text{lo}_i')=n^{-1}\sum_{i}\omega_{i}^{2}-(n^{-1}\sum_{i}s_i\omega_{i})^{2}>0$,
then $\tilde{\bm{\Sigma}}_{\nabla}$ is positive definite as
long as ${\bf X}$ is full rank.

Given these considerations, we propose the following regularised estimator 
\[
\mathring{\bm{\Sigma}}_{\nabla}=\alpha\hat{\bm{\Sigma}}_{\nabla}+(1-\alpha)\tilde{\bm{\Sigma}}_{\nabla},
\]
where $\alpha\in[0,1]$ determines the amount of regularisation. We
choose $\alpha=\text{min}(n_{e}/d^{2},1)$, where $n_{e}=(\sum_{i}\omega_{i})^{2}/\sum_{i}\omega_{i}^{2}$ is the
Kish's Effective Sample Size (ESS) implied by the weights. Given that $\tilde{\bm{\Sigma}}_{\nabla}$ is an inconsistent estimator in general, it is desirable that $\alpha\rightarrow 1$ as $n$ increases. By considering a simplified setting, SM \ref{app:ESSgrowth} proves that $\mathbb{E}(n_{e})$ is $O\{n\,\text{min}(1-\tau, \tau)\}$ when fitting extreme quantiles. If we assume that $d = O(n^{1/5})$, which is a relatively fast rate of basis growth for penalised regression splines (see e.g. \cite{wood2006generalized}, Section 5.2), we have that $\alpha = \text{min}[O\{n^{3/5}\text{min}(1-\tau, \tau)\}, 1]$.

\subsection{Asymptotic behaviour of $\mathbb{E}(n_e)$} \label{app:ESSgrowth}

Consider a simplified setting where $y_1, \dots, y_n$ are i.i.d. random variables and $\mu$ is a scalar. Without loss of generality, set $\tau\geq0.5$, $\sigma_0=1$, and notice that
\begin{align*}
\frac{\mathbb{E}(n_{e})}{n} & = \frac{\mathbb{E}\{(n^{-1}\sum_{i}\omega_{i})^{2}\}}{\mathbb{E}(n^{-1}\sum_{i}\omega_{i}^{2})} + O(n^{-1}) \geq \frac{\mathbb{E}(\omega)^2}{\mathbb{E}(\omega^{2})} + O(n^{-1}) \\ & \geq \frac{\tilde{\tau}^2 \text{Prob}\{y\geq Q(\tilde{\tau})\}^{2}}{(1-\tau)^2\text{Prob}[y\leq Q\{2(1-\tau)\}]+\tau^2\text{Prob}[y > Q\{2(1-\tau)\}]} + O(n^{-1}),
\end{align*}
for any $\tilde{\tau}\in(1-\tau,1)$, where $Q(u)=\mu+\lambda\text{log}\{u/(1-u)\}$ is the logistic quantile function. If we set $\tilde{\tau} = 0.5$ and evaluate $\mathbb{E}(n_{e})/n$ at $\mu = F^{-1}(\tau)$, where $F$ is the c.d.f. of $y$ under $\mathbb{P}$, we have that 
% We evaluate at \mu = F^{-1}(\tau) + O(\lambda) because of 
% F(mu^*) - F(mu_0) = O(\lambda).
\begin{align*}
\frac{\mathbb{E}(n_{e})}{n} & \geq \frac{1}{4}\frac{(1-\tau)^2}{(1-\tau)^2\tau+\tau^2(1-\tau)} + O(\lambda + n^{-1}) \propto 1-\tau + O\{(1-\tau)^2\} + O(\lambda + n^{-1}),
\end{align*}
As $n\rightarrow\infty$, consistency requires that $\lambda\rightarrow 0$, so $\mathbb{E}(n_{e})$ is $O\{n(1-\tau)\}$ when fitting extremely high quantiles ($\tau\approx 1$). Similar steps prove that $\mathbb{E}(n_{e})$ is $O(n\tau)$ when $\tau \approx 0$. 

\section{Details on the electricity forecasting application} \label{app:electrDetails}

We remove from the UK and French datasets all data between the 21st of December and the 4th of January (included) because, in an operational setting, forecasting electricity demand during this period requires manual intervention, as demand behaviour is anomalous relative to the rest of the year. For the same reason we exclude from the French dataset the period between the 26th of July and the 24th of August (included).

To forecast load one week ahead, we use the observed temperature over that week. Obviously future temperatures would not be available in an operational setting, and a forecast would be used instead. But using a forecast would add further uncertainty to the results of the comparison performed here, hence we prefer using observed temperatures. Week by week we predict the load for the next seven days, and then we re-fit all models using the newly observed values of load and temperature.

\end{appendices}

%\newpage

%\newpage

\end{document}